\RequirePackage[l2tabu, orthodox]{nag}
\documentclass{article}

\usepackage{amsmath}
\usepackage[a4paper]{geometry}
\usepackage{graphicx}
\usepackage{microtype}
\usepackage[colorlinks=false, pdfborder={0 0 0}]{hyperref}
\usepackage{cleveref}

\usepackage{amssymb}
\usepackage{lmodern}
\usepackage{amsthm}
\usepackage{xspace}
\usepackage{pxfonts}
\usepackage{algorithmicx}
\usepackage{algpseudocode}
\usepackage{authblk}
\usepackage[T1]{fontenc}

\theoremstyle{plain}
\newtheorem{theorem}{Theorem}[section]

\newtheorem{definition}[theorem]{Definition}

\theoremstyle{remark}

\newtheorem{corollary}[theorem]{Corollary}
\newtheorem{proposition}[theorem]{Proposition}
\newtheorem*{claim}{Claim}
\newtheorem{subclaim}{Subclaim}
\newtheorem{fact}[theorem]{Fact}
\newtheorem*{example}{Example}

\DeclareMathOperator{\UCT}{UCT}
\DeclareMathOperator{\proj}{proj}
\DeclareMathOperator{\CSP}{CSP}

\newcommand{\pq}{\ensuremath{pq}\xspace}
\newcommand{\SDm}{\textrm{SD}($\wedge$)\xspace}
\newcommand{\tempA}{{\mathbf A}}
\newcommand{\alg}[1]{{\mathbb #1}}
\newcommand{\algA}{{\mathbb A}}
\newcommand{\algB}{{\mathbb B}}
\newcommand{\algC}{{\mathbb C}}
\newcommand{\algD}{{\mathbb D}}
\newcommand{\algE}{{\mathbb E}}
\newcommand{\algF}{{\mathbb F}}
\newcommand{\algG}{{\mathbb G}}
\newcommand{\algP}{{\mathbb P}}
\newcommand{\algR}{{\mathbb R}}
\newcommand{\algS}{{\mathbb S}}
\newcommand{\inst}[1]{{\mathcal #1}}
\newcommand{\instI}{{\mathcal I}}
\newcommand{\instJ}{{\mathcal J}}
\newcommand{\subd}{\leq_{\text{sub}}}
\newcommand{\tuple}[1]{{\mathbf #1}}
\newcommand{\compose}[1]{{(#1)}}
\newcommand{\cover}{\varphi}
\newcommand{\congr}[1]{\alpha_{#1}}
\newcommand{\congrrel}[1]{\mathrel{\congr{#1}}}
\newcommand{\prop}{\sqsubseteq}
\newcommand{\abs}{\unlhd}
\newcommand{\lastcomponent}{{\mathcal R}}
\newcommand{\lastsets}{{\mathcal B}}
\newcommand{\sol}{e}

\begin{document}

\setlength{\pdfpageheight}{\paperheight}
\setlength{\pdfpagewidth}{\paperwidth}

\title{Weak consistency notions\\for all the CSP{s} of bounded width%
\thanks{Research partially supported by National Science Center under grant DEC-2011/01/B/ST6/01006.}} 

\author{Marcin Kozik}
\affil{Theoretical Computer Science Department\\The Faculty of Mathematics and Computer Science\\Jagiellonian University}
\renewcommand\Affilfont{\small}

\maketitle
\begin{abstract}
  The characterization of all the Constraint Satisfaction Problems of bounded width, 
  proposed by Feder and Vardi [SICOMP'98], was confirmed in [Bulatov'09] and
  independently in [FOCS'09, JACM'14]. 
  Both proofs are based on the (2,3)-consistency~%
  (using Prague consistency in [FOCS'09], directly in [Bulatov'09])
  which is costly to verify.

  We introduce a new consistency notion, Singleton Linear Arc Consistency~(SLAC),
  and show that it solves the same family of problems. 
  SLAC is weaker than Singleton Arc Consistency~(SAC) and thus the result answers the question from [JLC'13]
  by showing that SAC solves all the problems of bounded width.
  At the same time the problem of verifying weaker consistency~(even SAC) offers significant computational advantages 
  over the problem of verifying (2,3)-consistency which improves the algorithms 
  solving the CSPs of bounded width.
\end{abstract}

\section{Introduction}
  An instance of the Constraint Satisfaction Problem consists of variables and constraints.
  In the decision version of CSP the question is whether the variables can be evaluated in such a way
  that all the constraints, 
  often described as a relation constraining a sequence of variables,
  are satisfied.

  In a seminal paper~\cite{FV98} Feder and Vardi proposed to parametrize the problem by restricting the constraining relations
  allowed in instances. 
  More formally, for every finite relational structure $\tempA$~(called in this context {\em a template} or {\em a language})
  the CSP parametrized by $\tempA$, $\CSP(\tempA)$, is the CSP 
  restricted to instances with all the constraint relations taken from $\tempA$.
  
  Clearly, for any $\tempA$, the problem $\CSP(\tempA)$ is in NP and it is quite easy to construct
  relational structures which define NP-complete CSPs or CSPs solvable in polynomial time. 
  One of the main problems in the area is {\em the CSP Dichotomy Conjecture}~\cite{FV98}, which postulates 
  that for every $\tempA$ the problem $\CSP(\tempA)$ is NP-complete or solvable in a polynomial time.
  The CSP Dichotomy Conjecture remains open.

  The class of problems which can be expressed as a $\CSP(\tempA)$ is very rich; 
  it is easy to construct relational structure $\tempA$ such that $\CSP(\tempA)$ is
  2-colorability of graphs, 3-SAT, 3-Horn-SAT, or a problem of solving systems of linear equations
  in $\mathbb Z_2$. 
  The last problem on this list is a canonical example of a CSP with an {\em ability to count}.
  Feder and Vardi conjectured that all the CSPs which 
  do not have the ability to count are solvable by local consistency checking.
  Local consistency checking algorithms operate by 
  constructing a family of local solutions of an instance
  and enforcing some form of consistency on it.
  The CSPs solvable by such algorithms are {\em CSPs of bounded width}%
  \footnote{In this paper we use phrases "CSPs solvable by local consistency checking" and "CSPs of bounded width" interchangeably.},
  and the conjecture of Feder and Vardi is the {\em the bounded width conjecture}.%
  \footnote{Not to be confused with the dichotomy conjecture of Feder and Vardi.}

  A breakthrough in the research on the parametrized CSPs appeared with
  an introduction of the {\em algebraic approach}~\cite{duality1,duality2}.
  This approach is based on the Galois correspondence
  between relational structures and algebras~\cite{GaloisOLD}.
  At it's heart lies a method of associating 
  algebras to templates in such a way that the computational properties of CSP~(parametrized by the template)
  correspond to high-level algebraic properties of the algebra.

  An algebraic approach allowed to restate and formalize~\cite{algebraicBW} the bounded width conjecture.
  The restated conjecture postulated that 
  a template has bounded with if and only if the associated algebra generates a variety which
  is congruence meet semi-distributive.
  This conjecture, and consequently the bounded width conjecture, 
  was confirmed by two independent algebraic proofs~\cite{FOCSBW,BW,BulBW}.

  CSPs of bounded width form a big class of problems and appear naturally in many areas.
  To name one such connection: Guruswami and Zhou conjectured~\cite{RobustConj} that the 
  class of CSPs which admit {\em robust approximating algorithm}~(i.e. a poly-time algorithm
  that, if $1-\varepsilon$ fraction of constraints can be satisfied in an instance,
  provides a ``solution'' satisfying $1-g(\varepsilon)$ fraction of constraints 
  and $g(\varepsilon)\rightarrow 0$ as $\varepsilon\rightarrow 0$) coincides with the 
  class of CSPs of bounded width.
  This conjecture was confirmed in~\cite{DKRobust,RobustSTOC}.

  The study of consistency notions for CSP predates 
  the parametrized approach of Feder and Vardi.
  In many practical applications a consistency checking algorithm is used
  to quickly rule out some instances with no solutions with less regard 
  for the existence of solution if the instance cannot be ruled out.
  Considered instances are usually large and very often sparse 
  and thus require algorithms which work fast and do not disturb the structure of the instance.
  For an overview of consistency notions and algorithms we refer the reader to e.g.~\cite{Debruyne97somepracticable, Bodirsky2010, SAC}. 
  
  The proof~\cite{FOCSBW,BW} of the bounded width conjecture required $(k,l)$-con\-sis\-ten\-cy~%
  (compare Definition 3.1 in~\cite{BWColapse}) with $k,l$ dependent on the maximal arity of relations in the template.
  The main theorem of~\cite{BWColapse} shows that $(2,3)$ minimality suffices for all the templates.
  On the other hand both proofs rely on a technical notion of Prague instance, which can be viewed as another
  consistency notion. 
  The proof of~\cite{RobustSTOC} required another consistency definition: a definition of a weak Prague instance
  for templates with binary constraints~(Definition 3.3 in~\cite{RobustSTOC}).
  All these results, as well as~\cite{BulBW}, even in restriction to templates with binary constraints
  require some version of $(2,3)$-consistency which is often discredited in practical applications.

  Singleton arc consistency is, on the other hand, a  well established consistency notion 
  easier to verify than $(2,3)$-consistency.
  Moreover an algorithm can verify SAC without distorting the structure of the instance~%
  (unlike the algorithms verifying $(2,3)$-consistency).
  In~\cite{SAC} authors discuss applicability of SAC to CSPs of bounded width 
  and ask if every CSP of bounded width is solvable by SAC.

  We answer this question in positive by introducing a new consistency notion,
  Singleton Linear Arc Consistency, which is 
  weaker than SAC and showing that it solves all the CSPs solvable by local consistency checking.
  The notion is not weaker than the one used in~\cite{RobustSTOC}
  to establish the conjecture of Guruswami and Zhou, but, at the cost of complicating proofs, 
  it could be adjusted to imply that result as well.
  The definition of SLAC is motivated by the work in~\cite{majorityLinDATA,Bodirsky2010}
  and provides a new insight into bounded width templates 
  which can be important when attempting to construct algorithms 
  solving wider classes of CSPs~(with the dichotomy conjecture as a long term goal);
  or attempting a fine-grained classification of CSPs along the lines of 
  e.g.~\cite{majorityLinDATA,nuLinDATA}.

  The paper is organized as follows.
  In the next section we describe a number of consistency notions and define the Singleton Linear Arc Consistency.
  In section~\ref{sect:basicinstances} we 
  state the main theorem of the paper and provide a rough sketch of the proof.
  Section~\ref{sect:algebra} contains definitions of algebraic notions and collects established tools in the algebraic approach to CSP.
  In section \ref{sect:patterns} we rephrase the consistency notions from section~\ref{sect:basicinstances} into a language
  which allows us to work with the instances of CSP.
  In section~\ref{sect:result} we return to the main theorem of the paper and present a more in-depth sketch of the proof
  deferring most of the actual reasoning to the very last sections of the paper.
  In sections~\ref{sect:toolsabs},~\ref{sect:toolsnoabs} we derive tols 
  responsible for the two cases in the proof of the main theorem.
  Finally in sections~\ref{sect:noabs},~\ref{sect:abs} we finsh the proof of the main result.
  The last section contains acknowledgments.

\section{Local consistency notions for CSPs} \label{sect:consistency}
  The purpose of local consistency checking is to eliminate,
  at a low computational cost,
  instances with no solutions. 
  However if an algorithm checking a local consistency notion stops without deriving a contradiction
  the instance does not need to have a solution.
  In the parametrized approach to  CSP the template 
  has bounded width if some local consistency notion derives a contradiction for every non-solvable instance. 
  
  Formally 
  {\em a template} of a CSP is a fixed, finite relational structure usually denoted by $\tempA$, and
  {\em an instance of  CSP over template $\tempA$} consists of a set of variables and a set of constraints
  which are pairs: $((x_1,\dotsc,x_n),R)$ where $x_i$ is a variable and $R$ is an $n$-ary relation in $\tempA$.
  {\em A solution of an instance} is a function $f$ sending variables to the universe of $\tempA$,
  usually denoted by $A$,
  in such a way that for every constraint $((x_1,\dotsc,x_n),R)$ we have $(f(x_1),\dotsc,f(x_n))\in R$.

  There exists an extensive literature on various local consistency notions
  used in CSP~\cite{Debruyne97somepracticable,Bodirsky2010,SAC} etc.
  In the remaining part of this section we present only these 
  which are directly pertinent to our approach.
  This is done using the notion of a DATALOG program.

  \subsection{DATALOG programs}
    A DATALOG program derives new facts about a relational structure using a set of rules.
    It operates in the languages of this structure~%
    (the relations/predicates%
    \footnote{We abuse the notation by identifying predicates with the relations in the structure.} 
    of that language are called {\em extensional database} or EDBs) 
    enhanced by a number of auxiliary predicates/relations~%
    (called {\em intensional database} or IDBs).

    A DATALOG rule has a {\em head} which is a single IDB on an appropriate number of  variables
    and the {\em body} which is a sequence on IDBs and EDBs.
    During the execution of the program the head IDB is updated
    whenever the body of the rule is satisfied. 
    The computation ends when no relation can be updated, 
    or when the goal predicate is reached.

    \begin{example}
      The following DATALOG program operates on digraphs.
      The edge relation of a digraph~(denoted by $E$) is the single binary EDB in the program.
      The program uses two IDBs: $ODD$ and $GOAL$~(where $GOAL$ is the goal predicate) and 
      verifies whether the digraph has a directed cycle of odd length.
      \begin{align*}
        ODD(x,y) &\coloneq E(x,y) \\
        ODD(x,v) &\coloneq E(x,y) \wedge E(y,z) \wedge ODD(z,v) \\
        GOAL &\coloneq ODD(x,x)
      \end{align*}
      The program computes the relation $ODD$ and fires the $GOAL$ predicate
      when a directed circle of odd length is found.
    \end{example}
    
  \subsection{Arc consistency~($1$-consistency)}\label{sect:ac}
    Arc consistency is one of the most basic local consistency notions. 
    For the purpose of this paper we present it in two ways: using a DATALOG program
    and defining an iterative procedure.
    The notion is sometimes~\cite{DalmauWidth1} called {\em generalized arc consistency} 
    as it is adjusted to work with relations of all arities.
    
    A DATALOG program verifying arc consistency for a template $\tempA$ has
    an IDB for each subset of the universe of $\tempA$. 
    Whenever $B(x)$~(for a variable $x$ and $B\subseteq A$) is derived 
    we understand that the program established that 
    in every potential solution the variable $x$ is evaluated into $B$.
    
    The derivation itself proceeds in steps: 
    in each step we focus a single constraint $((x_1,\dotsc,x_n),R)$  
    and use already established facts~(i.e. IDBs) about the variables
    \begin{equation*}
      A'_1(x_1), A''_1(x_1),\dotsc,A'_2(x_2),\dotsc
    \end{equation*}
    to derive that $x_i$ needs to be evaluated in the appropriate projection of%
    \footnote{Note that a variable need not appear in any of the IDBs, to address this purely technical  problem 
      we set the intersection of the empty family of sets to be $A$.}
    \begin{equation*}
      R\cap \big(\bigcap_jA_1^{(j)}\times\dotsb\times\bigcap_jA_n^{(j)}\big). \tag{\ddag}
    \end{equation*}
    A DATALOG rule realizing this derivation is 
    \begin{align*}
      B(x_i) &\coloneq R(x_1,\dotsc,x_n) \wedge \\
      &\qquad A_1'(x_1) \wedge A_1''(x_1)\wedge \dotsc \wedge \\
      &\qquad A_2'(x_2) \wedge A_2''(x_2)\wedge \dotsc \wedge \\
      &\qquad \dotsc \\
      &\qquad A_n'(x_n)\wedge A_n''(x_n) \wedge\dotsc
    \end{align*}
    where $B$ is a projection of $(\ddag)$ to the $i$--th coordinate.
    To construct a DATALOG program verifying arc consistency for a fixed template $\tempA$
    we collect all such rules, 
    and set the goal predicate to be the predicate associated with an empty set.
    Whenever this goal predicate is reached we will say that the program {\em derived a contradiction}.

    Alternatively it can be described be a procedure updating, 
    for every variable $x$, a special unary constraint $C_x$.
    The idea behind these constraints is to store information derived about the variable. 
    More formally:
    \newcommand{\pushcode}[1][1]{\hskip\dimexpr#1\algorithmicindent\relax}
    \begin{algorithmic}
      \For{every variable $x$ in $\instI$}
        \State add to $\instI$ a new special constraint $C_x := (x,A)$
      \EndFor
      \Repeat
        \For{every variable $x$ in $\instI$}
          \State $B$ := $A$
          \For{every value $a\in A$}
            \For{each constraint $((x_1,\dotsc,x_n),R)$ with $x_i=x$}
              \State let $R'\subseteq R$ contain all tuples respecting $C_{x_j}$'s
              \If{no $(a_1,\dotsc,a_n)$ in $R'$ has $a_i=a$}
                \State remove $a$ from $B$
              \EndIf
            \EndFor
          \EndFor
          \State $C_x$ := $(x,B)$
        \EndFor
      \Until{none of the $C_x$ changed}
    \end{algorithmic}
    For a constraint $((x_1,\dotsc,x_n),R)$ a tuple $(a_1,\dotsc,a_n)\in R$ respects $C_{x_j} = (x_j,A_{x_j})$ if
    $a_j\in A_{x_j}$.
    The algorithm {\em derives a contradiction} when one of the special constraints is of the form 
    $C_x = (x,\emptyset)$ i.e. has an empty constraint relation.

    Note that if the instance has a solution sending variable $x$ to $a\in A$
    and the Datalog program derives $B(x)$ then we have $a\in B$,
    and in the procedural version if $C_x = (x,B)$ at some step, then $a\in B$.
    Therefore deriving an empty predicate on a variable~(or a special constraint with an empty set)
    implies that the variable has no possible values
    and that the instance has no solution.
    
    \begin{example}
      Let the template $\tempA$ have universe $A = \{0,1\}$ and a single binary relation $\neq$.
      To avoid confusion we denote the IDB associated with $\{0\}$ subset~(resp. $\{1\}$ subset) by
      $\cdot = 0$~(resp. $\cdot = 1$).
      The following DATALOG program verifies $1$-consistency for $\tempA$:
      \begin{align*}
        \emptyset(x) &\coloneq  x=0 \wedge x=1, \\
        x=0 &\coloneq x\neq y \wedge y=1, \\ 
        y=0 &\coloneq x\neq y \wedge x=1, \\ 
        x=1 &\coloneq x\neq y \wedge y=0, \\
        y=1 &\coloneq x\neq y \wedge x=0.
      \end{align*}
      Note that all the rules are necessary as a Datalog program disregards 
      commutativity of $\neq$.
      
      On the instance with variables $x,y,z$ and constraints $x\neq y, y\neq z, z\neq x$,
      this program derives no 
      contradiction despite the fact that the instance has no solution.
    \end{example}
    This example shows that the template $(\{0,1\},\neq)$ is not solvable by arc consistency.
    As it is solvable by a different local consistency notion we conclude that
    the arc consistency is not strong enough to verify all the CSPs of bounded width.
    In fact all the CSPs solvable by arc consistency are characterized in~\cite{DalmauWidth1}.

    Note that, for  a fixed template,
    arc-consistency can be computed quite quickly i.e. in the time linear with respect to the number of constrains.
  
  \subsection{Path consistency~($(2,3)$-consistency)}
    The results establishing~\cite{BulBW,FOCSBW,BW} the bounded width conjecture of Feder and Vardi~\cite{FV98}
    were using higher consistency notions. 
    The paper of Barto~\cite{BWColapse} improves these results using slightly different concepts.
    We will not define these notions since they are not directly connected to 
    this paper. 
    We just mention that the result~\cite{BWColapse} showed that
    if the template has bounded width then every $(2,3)$-minimal instance has a solution.

    When restricted to templates with binary relations all the above results~\cite{BulBW, FOCSBW, BW, BWColapse}
    require  some form of $(2,3)$-consistency also known as path-consistency. 
    A DATALOG program verifying $(2,3)$-consistency, for a binary template,
    is constructed in a way similar to the program for arc consistency, 
    except for two differences:
    \begin{itemize}
      \item it introduces a binary IDB for each subset of $A^2$~(instead of unary ones for subsets of $A$), and
      \item it allows derivation rules with body on three variables~%
        (arc-consistency on binary relations would have at most two variables in the body of every DATALOG rule).
    \end{itemize}
    Even when restricted to binary constraints SLAC, or SAC, is strictly weaker than $(2,3)$-consistency.
    For constraints of arbitrary arities both notions are strictly weaker than $(2,3)$ minimality.

    Verifying $(2,3)$-consistency is inefficient for many reasons: the computational complexity is cubic 
    with respect to the number of variables~\cite{Debruyne97somepracticable}. 
    More importantly the instance of the constraint satisfaction program is loosing its structure.
    In many practical applications the CSP instances are sparse --- 
    establishing arc consistency do not change the structure of an instance, 
    however establishing $(2,3)$-consistency on any~(even sparse) instance would essentially
    add a constraint for every pair of variables loosing a significant computational advantage.

    The proofs in~\cite{FOCSBW, BW, BWColapse} were using yet another, technical, consistency notion:
    Prague consistency. 
    This notion, although interesting theoretically, is not easy to verify and 
    did not lead to algorithms more efficient then those verifying $(2,3)$-consistency.

  \subsection{Singleton Arc Consistency~(SAC)}\label{sect:sac}
    Singleton arc-consistency is a notion stronger than arc consistency, but weaker than path consistency.
    That means that an instance which is path consistent is necessarily SAC,
    and that every SAC instance is arc consistent. 
    The reverse implications do not hold.

    Let us fix an instance $\instI$. 
    SAC is defined using a procedure similar to the one defining AC. 
    This time the constraints $C_x$ are updated by running arc consistency with {\em the value of $x$ fixed to an arbitrary $a$}.
    Formally:
    \begin{algorithmic}
      \For{every variable $x$ in $\instI$}
        \State add to $\instI$ a new special constraint $C_x := (x,A)$
      \EndFor
      \Repeat
        \For{every variable $x$ in $\instI$}
          \State $B$ := $A$
          \For{every value $a\in A$}
            \State run AC on $\instI$ restricted by $C_y$'s and  $(x,\{a\})$
            \If{AC derived a contradiction}
              \State remove $a$ from $B$
            \EndIf
          \EndFor
          \State $C_x$ := $(x,B)$
        \EndFor
      \Until{none of the $C_x$ changed}
    \end{algorithmic}
    Restricting $\instI$ by $C_y$'s means substituting, in every constraint $((x_1,\dotsc,x_n),R)$, 
    the relation $R$ with 
    $R\cap\prod_{i=1}^n A_{x_i}$ where the $A_{x_i}$'s are taken from special constraints i.e. $C_{x_i}=(x_i,A_{x_i})$. 
    The additional constraint $(x.\{a\})$ fixes value of $x$ to $a$.
    The algorithm derives a contradiction in the same way the algorithm for AC did.

    The algorithm above, while conceptually simple, does not provide the best possible running time.
    Using the fact that the arc consistency, in the inner most loop, is computed multiple times on more or less the same instance
    it is possible to verify  SAC in time bounded
    by a constant times the number of constraints times the number of variables~\cite{SACCompl}.
    More importantly SAC does not alter the structure of the instance and therefore runs efficiently on 
    sparse instances with many variables.

    SAC has been studied~\cite{SAC} but there was no characterization of the set of templates
    solvable by SAC. 
    As Singleton Linear Arc Consistency is weaker then SAC the main theorem of this paper
    shows that SAC solves all the CSPs of bounded width.
    SLAC offers some computational advantages over SAC, however it is not clear if it allows 
    to construct an algorithm with better worst-case time complexity than the algorithm mentioned above.

  \subsection{Linear Arc Consistency~(LAC)}
    This consistency notion is a weaker version of arc consistency and originates in~\cite{Bodirsky2010,majorityLinDATA}.

    In order to define the notion we take the DATALOG program for arc consistency from section~\ref{sect:ac}
    and remove from it all the rules with more than one IDB in the body.
    Linear Arc Consistency is the consistency verified by this program.
    As the  program has fewer derivation rules than the original one 
    it computes less information about the instance.
    
    On the other hand
    LAC can be verified in NL~(as it essentially reduces to reachability for directed graphs),
    while AC solves some of the P-complete CSPs~(e.g. Horn-SAT) 
    and we cannot expect to put it into a low complexity class.

  \subsection{Singleton Linear Arc Consistency~(SLAC)}\label{sect:slac}
    This is the main consistency notion of this paper.
    It is stronger than AC, but weaker than SAC.
    In Theorem~\ref{thm:mainsimple}
    we show that SLAC solves all the constraint satisfaction problems of bounded width.

    The Singleton Linear Arc Consistency 
    is defined by an algorithm almost identical to the algorithm for SAC in section~\ref{sect:sac}; 
    then only difference lies in the inner most loop where~(for SLAC) we evaluate LAC instead of AC.
    Note that it is important~(unlike in the case of SAC) that the LAC is evaluated in a restricted instance.
    Still the cost of excluding a single candidate for a variable is~(unlike for SAC) in NL.

\section{Parametrized CSP and the main theorem}\label{sect:basicinstances}
  
  The following corollary, of a more technical Theorem~\ref{thm:mainsimple},
  is the easiest way to state the main result of the paper.
  \begin{corollary}\label{cor:simple}
    SLAC derives a contradiction {\em on every} unsolvable instance over a bounded width template.
  \end{corollary}
  \noindent As SLAC cannot derive a contradiction on a solvable instances we immediately get
  that SLAC solves the CSP for every template of bounded width.

  However, in order to prove Theorem~\ref{thm:mainsimple} and thus the corollary, we use algebraic characterizations which 
  require a number of standard reductions.
  First reduction, already present in~\cite{FV98}, allows us to add to $\tempA$ relations pp-definable in $\tempA$.
  A relation is {\em pp-definable} in $\tempA$ if it can be defined using relations in $\tempA$,
  conjunction and existential quantifiers. 
  Adding a pp-definable relation to a template does not increase the computation complexity of the corresponding CSP.
  More importantly other properties, 
  like being solvable by local consistency checking, 
  are preserved under this construction.
  Throughout the paper we assume that all the pp-definable unary relations are already present in $\tempA$.

  Further standard reductions allow us to focus on $\tempA$ which are {\em cores} that is
  relational structures such that every endomorphism of $\tempA$~(i.e. a homomorphism from $\tempA$ to $\tempA$)
  is a bijection. 
  The final standard reduction allows us  consider only {\em rigid cores} i.e. relational structures 
  for which identity is the only endomorphism~\cite{FV98,duality1,duality2} .
  
  For technical reasons the instances we consider throughout the paper have, 
  for every variable $x$, a constraint of the form $(x,A_x)$ 
  where $A_x$ is a subset of $A$ pp-definable in $\tempA$.
  For historical reasons the sets $A_x$ are sometimes called potatoes.
  Note that every instance can be turned into such an instance by adding dummy constraints of the form $(x,A)$.
  
  \subsection{$1$-consistent instances}
    \begin{definition}
      An instance is $1$-consistent if, for every constraint $((x_1,\dotsc,x_n),R)$ in the instance,
      the projection of $R$ to the $i$-th coordinate is equal to $A_{x_i}$.
    \end{definition}
    Note that the Datalog program for arc consistency run on a $1$-consistent instance
    would immediately derive $A_x(x)$ for every variable $x$ 
    and would not derive any further information.
    Moreover, for any instance $\instI$, we can use AC to construct a $1$-consistent instance.
    If AC stops without deriving a contradiction then, for every variable $x$ of $\instI$ we have 
    a special constraint of the form $C_x = (x,A_x)$ for some non-empty $A_x$.

    After restricting the instance by $C_x$'s~(equivalently to $A_x$'s), 
    as in section~\ref{sect:sac}, we obtain 
    a new instance $\instI'$.
    The instance $\instI'$ has exactly the same set of solutions as $\instI$ and is $1$-consistent~%
    (since the derivation of AC stopped).
    Moreover the new relations appearing in the constraints are pp-definable in the relational structure.
    We sometimes call such $\instI'$ {\em a $1$-consistent subinstance of $\instI$}.

  \subsection{SLAC instances and the main theorem}\label{sect:mainsimple}
    Comparing the procedural version of AC and SLAC we immediately 
    obtain that the SLAC verifies~(among other things) arc consistency of the instance.
    If the SLAC algorithm stops, on an instance $\instI$,  with special constraints of the form $C_x=(x,A_x)$
    we can restrict $\instI$ by $C_x$'s --  a result of such a restriction is a {\em SLAC instance} .

    This instance has the same set of solutions as the original instance, 
    it is $1$-consistent and SLAC run on such an instance would immediately derive $C_x=(x,A_x)$
    and would not derive any further information~%
    (this property is sometimes used as an equivalent definition of a SLAC instance).

    The main theorem of this paper claims solutions for SLAC instances:
    \begin{theorem}\label{thm:mainsimple}
      Let $\tempA$ be a template which is a rigid core. 
      If $\CSP(\tempA)$ has bounded width,
      then every SLAC instance in $\CSP(\tempA)$ has a solution.
    \end{theorem}       
    \begin{proof}[Proof of Corollary~\ref{cor:simple}]
      Let $\instI$ be an instance over $\tempA$~%
      (of bounded width) 
      such that the SLAC algorithm does not derive a contradiction on $\instI$.

      Restricting by $C_x$'s, derived by SLAC, we obtain a SLAC instance over $\tempA$.
      Let $\tempA'$ be a core of $\tempA$ via a homomorphism $h$.
      Applying $h$ to $\instI'$ we obtain a SLAC instance over $\tempA'$,
      which is also an instance over a rigid core of $\tempA$.

      Since $\tempA$ had bounded width so does the rigid core of $\tempA$ and 
      we can use Theorem~\ref{thm:mainsimple} to establish a solution to the image of $\instI'$ under $h$.
      Since $h$ is an endomorphism of $\tempA$ this solution works for $\instI$ as well.
    \end{proof}

    We present a short sketch of the proof of Theorem~\ref{thm:mainsimple};
    this sketch is extended in section~\ref{sect:result}
    and the full proof can be found in the last two sections of the paper.
    
    The general idea of the proof is to start with an arbitrary SLAC instance $\instI$
    and show that, choosing appropriate pp-definable $A'_x\subseteq A_x$,
    the restriction of $\instI$ to $A'_x$ is a proper SLAC subinstance of $\instI$. 
    If this can be accomplished then, in a finite number of steps, we arrive at a SLAC instance such that every $A_x$ has
    one element, and such an instance clearly defines a solution.

    Finding proper $A'_x$ for an instance splits into two cases depending on the algebraic structure of the instance.
    Either we can find $A'_x$ which {\em absorbs}~(compare section~\ref{sect:absdef}) $A_x$,
    or the lack of absorption implies enough "rectangularity" of the constraints that $A'_x$ can be chosen 
    more arbitrarily.
    
    The next section contains definitions of algebraic concepts necessary to make these statements precise.

\section{Algebraic notions and tools}\label{sect:algebra}
  First we introduce general algebraic notions which allow
  us to define the Galois correspondence.
  \subsection{Basic algebraic notions}
    {\em An algebra}, usually denoted by $\algA$,  is a set $A$~(the {\em universe} of the algebra)
    together with a list of functions of arbitrary~(but finite) arity from $A$ to $A$~(the {\em operations} of the algebra).
    {\em A signature} of an algebra is a list of arities of its operations.

    Let $\algA$ be an algebra, a set $B\subseteq A$ is {\em a subuniverse} of $\algA$ if it
    is closed under every operation of $\algA$. 
    In such a case the algebra with universe $B$ and operations obtained from the operations of $\algA$ 
    by restricting to arguments from $B$ is a {\em  subalgebra} of $\algA$ denoted $\algB\leq\algA$.

    Let $\algA_i$ be a list~(or, more generally, an indexed family) of algebras in the same signature, the algebra $\prod_i\algA_i$
    is the algebra in the same signature; 
    its universe is a cartesian product of
    universes of $\algA_i$'s and the operations are evaluated coordinatewise.
    If $\algA$ is an algebra then $\algA^n$ is the $n$-th cartesian power of $\algA$ i.e.
    a product of $n$ copies of $\algA$.

    An equivalence relation $\alpha$ on the universe of an algebra $\algA$ is {\em a congruence}
    if it is a subalgebra of $\algA^2$. 
    In such a case  the algebra $\algA/\alpha$ is well defined i.e. the operations in the quotient do not depend 
    on the choice of representatives of the equivalence class.
    For any $a\in\algA$ by $a/\alpha$ we mean the equivalence block of $\alpha$ containing $a$.

    An algebra $\algA$ is {\em simple} if its only congruences are $0_{\algA} = \{(a,a) : a\in\algA\}$
    and $1_{\algA} = \{(a,b) : a,b\in\algA\}$.
    If $\alpha$ and $\beta$ are congruences on $\algB$ then the largest
    equivalence contained in both is a congruence denoted by $\alpha\wedge\beta$ and 
    the smallest congruence containing $\alpha$ and $\beta$ is denoted by $\alpha\vee\beta$.
    In general, for any binary relation on an algebra, the congruence {\em generated} by this relation
    is the unique smallest congruence containing this relation.

    If $\algA\leq\prod_{i\in I}\algA_i$ then $\pi_i$ is the congruence identifying tuples of $\algA$ 
    with the same element on the $i$-th coordinate;
    and for any $J\subseteq I$ the algebra $\proj_J(\algA)$ is obtained from $\algA$ by taking tuples in $\algA$  
    and projecting 
    them to the coordinates from $J$.
    If $\algA\leq\prod_{i}\algA_i$ and for every $i$ $\proj_i(\algA)=\algA_i$ then $\algA$ is {\em subdirect}
    and we denote this fact by $\algA\subd\prod_i\algA_i$.
    
    If $\algC\leq\algA^2$ then by $\algC^\compose{n}$ we denote the subalgebra of $\algA^2$ 
    which is obtained by composing $\algC$ with itself $n$-times.
    If $\algC\leq\algA\times\algB$ and $\algA'\leq\algA$ then $\algA'+\algC$ is 
    a subalgebra of $\algB$ containing $b$ iff $\exists\ a\in\algA' : (a,b)\in\algC$.
    For $\algB'\leq\algB$ we define $\algB'-\algC$ analogously. 
    
    A homomorphism from $\algA$ to $\algB$~(where $\algA$ and $\algB$ have the same signature)  is a map from $A$ to $B$ 
    which commutes with operations of $\algA$ and $\algB$. 
    A bijective homomorphism is {\em an isomorphism}.
    
    {\em A term} in a given signature is a formal description of a composition of operations of an algebra in
    this signature. 
    For every algebra in this signature a term defines a {\em term operation}. 

    A class of algebras in the same signature closed under taking products, homomorphic images and subalgebras 
    is called {\em a variety}.
    The smallest variety containing an algebra $\algA$ is the variety {\em generated} by $\algA$.
    
    An algebra $\algA$ is idempotent if every one-element subset of its universe is a subalgebra.
    In such an algebra a congruence block is a subalgebra;
    moreover every algebra in a variety generated by an idempotent algebra is idempotent.
    \subsubsection{Absorption}\label{sect:absdef}
      The notion of {\em absorption} and {\em absorbing subalgebras} appeared first in~\cite{FOCSBW}
      and plays a crucial role in many recent developments in the algebraic approach to CSP.

      Let $\algB\leq\algA$ be idempotent algebras, we say that $\algB$ {\em absorbs} $\algA$~(denoted $\algB\abs\algA$)
      if there exists a term operation $t$ in $\algA$ such that
      \begin{equation*}
        t(a_1,\dotsc,a_n)\in\algB \text{ whenever } |\{i:a_i\notin\algB\}|\leq 1.
      \end{equation*}
      The following easy consequences of the definition will be useful:
      \begin{itemize}
        \item if $\algB\abs\algA$ and $\algC\leq\algA$ then $(\algB\cap\algC)\abs\algC$;
        \item if $\algC\subd\algA\times\algB$ and $\algA'\abs\algA$ then $\algA'+\algC\abs\algB$;
        \item  if $\algB\abs\algA, \algB'\abs\algA'$ are algebras in the same signature,
          then both absorptions can be witnessed by a common term.
      \end{itemize}

  \subsection{The Galois correspondence and basic reductions}
    At the heart of the algebraic approach to CSP lies a correspondence~\cite{GaloisOLD,duality1,duality2} 
    between relational structures and algebras.
    
    To each template $\tempA$ we associate an algebra $\algA$: 
    a $k$-ary function $f$ is an operation of $\algA$ if and only if every relation in $\tempA$
    is a subuniverse of appropriate power of the algebra $(A,f)$.
    Such an operation is called a {\em polymorphism} of $\tempA$.

    For such an $\algA$ all the finite subpowers of $\algA$ are exactly the relations pp-definable in $\tempA$
    which indicates why $\algA$ captures the complexity of CSP defined by $\tempA$.
    In the remaining part of the paper, all the relations
    appearing in constraints are, at the same time, subalgebras of powers of $\algA$.
    Note that, for any instance, the set of solutions of this instance also forms a
    subalgebra of a power of $\algA$.
  
    By standard reductions listed in section~\ref{sect:basicinstances} we restrict our reasoning to templates which 
    are rigid cores. 
    The algebras associated to such templates by the Galois correspondence are idempotent. 

    By results of~\cite{FOCSBW,BW,BulBW} a rigid core template $\tempA$ defines a CSP solvable by local consistency checking
    if and only if
    the associated algebra $\algA$ generates a variety such that for any algebra $\algB$ in this variety
    and any $\alpha,\beta,\gamma$ congruences of $\algB$ if $\alpha\wedge\beta = \alpha\wedge\gamma$ then
    $\alpha\wedge\beta = \alpha\wedge(\beta\vee\gamma)$.
    We call such an algebra $\algA$ an \SDm algebra, and we call a rigid core template an \SDm template 
    if the algebra associated to it is \SDm.

    A notion of a Taylor algebra~(which appears in Theorem~\ref{thm:abs}) is not required in this paper,
    all that we need to know is that every \SDm algebra is Taylor.

  \subsection{Algebraic tools}
    In this section we cite some established algebraic tools which were developed
    with the algebraic approach to CSP in mind.

    From this point on we assume that all the algebras are idempotent and all the templates 
    are rigid cores.
    \begin{fact}\label{fact:linkness}
      Let $\algC\subd\algA\times\algB$:
      \begin{itemize}
        \item if $\alpha$ is a transitive closure of the relation on $\algA$ containing all pairs 
          \begin{equation*}
            (a,a'): \exists\, b\in \algB\, (a,b),(a',b)\in\algC
          \end{equation*}
          then $\alpha$ is a congruence on $\algA$;
        \item similarly for $\beta$ defined dually on $\algB$;
        \item if $\algC'\subd\algA\times\algB$ and $\algC'\abs\algC\subd\algA\times\algB$ then $\alpha$'s and $\beta$'s 
          defined by $\algC$ and $\algC'$ coincide.%
          \footnote{The property is proved by a standard absorption argument found in e.g.~\cite{FOCSBW}.}
      \end{itemize}
    \end{fact}
    In fact, we call $\algC\subd\algA\times\algB$ {\em linked} if 
    $\pi_1\vee\pi_2 = 1_{\algC}$ or, equivalently
    $\alpha$~(or $\beta$) defined
    for $\algC$ in Fact~\ref{fact:linkness} are $1_\algA$~($1_\algB$ respectively).

    Note that if $\algC\subd\algA\times\algB$ is a graph of a bijection~(equivalently $\alpha,\beta$ in Fact~\ref{fact:linkness}
    are $0_{\algA},0_{\algB}$~respectively) then $\algC$ is actually a graph of an isomorphism between $\algA$ and $\algB$.

    The following theorem first appeared as Theorem III.6~\cite{LICScyclic}
    \begin{theorem}\label{thm:abs}
      Let $\algC\subd\algA\times\algB$ be a Taylor algebra.
      If $\algC$ is linked  then 
      \begin{itemize}
        \item $\algC = \algA\times\algB$, or
        \item $\algA$ has a proper absorbing subalgebra or
        \item $\algB$ has a proper absorbing subalgebra.
      \end{itemize}
    \end{theorem}
    \noindent And the following corollary easily follows from it:
    \begin{corollary}\label{cor:MAS}
      Under the assumptions of Theorem~\ref{thm:abs} there exist $\algA'\abs\algA$, $\algB'\abs\algB$
      such that $\algA'\times\algB'\leq\algC$.
    \end{corollary}
    \noindent The following proposition is, in essence, Theorem 6 of~\cite{nuLinDATA}.
    \begin{proposition}\label{prop:absconst}
      Let $\algA$ be an algebra and $\algR,\algS\subd\algA^n$ such that $\algR\abs\algS$, 
      and for every $a\in \algA$ the constant tuple $(a,\dotsc,a)$ belongs to $\algS$.
      Then $\algR$ contains at least one constant tuple.
    \end{proposition}

\section{Instances and patterns}\label{sect:patterns}
  In order to apply the algebraic tools to SLAC, or for that matter $1$-consistent, instances we need a better way 
  to handle such instances.
  In the following section we introduce patterns~(which are essentially CSP instances with added information)
  and a natural notion of a homomorphism between instances in order to capture these consistency notions. 

  Throughout the paper we work with instances of CSP over a fixed, finite template.
  Every instance of CSP over such a template can be equivalently viewed as a relational structure 
  in the language of the template.
  This was already the case when we considered a DATALOG program computing on an instance of 
  a CSP. 

  Having two instances of the same CSP, say $\instI$ and $\instJ$, and $\cover$ 
  a function mapping variables of $\instI$ 
  to variables of $\instJ$ we say $\cover$  is a homomorphism if 
  for every constraint $((x_1,\dotsc,x_n),\algC)$ in $\instI$ there is a constraint
  $((\cover(x_1),\dotsc,\cover(x_n)),\algC)$ in $\instJ$.
  This is exactly the same as saying that $\cover$ is a homomorphism between relational structures 
  which $\instI$ and $\instJ$ essentially are.

  \subsection{Path patterns}
    We begin with a basic definition.
    \begin{definition}
      {\em A step} is an instance of a CSP with a single constraint~%
      (in which every variable appears once) with two selected variables. 
      We denote a step by $(i,((x_1,\dotsc,x_n),\algC),j)$ where $((x_1,\dotsc,x_n),\algC)$ is the constraint and 
      the selected variables are $x_i$ {\em the beginning} of the step and $x_j$ {\em the end}.
    \end{definition}
    \noindent We define path-patterns next.
    \begin{definition}
      {\em A path-pattern} is an instance of CSP constructed from a sequence of steps~(on pairwise disjoint sets of variables)
      by identifying each step's end variable with next step's beginning variable.
      The beginning variable of the path-pattern is the beginning variable of the first step,
      and the end variable of the path-pattern is the end variable of the last step.

      {\em A subpattern of a path-pattern} is a path-pattern defined by a substring of the sequence of steps.

      A path-pattern $p$ is {\em in the instance $\instI$} if there exists a homomorphism~%
      (usually denoted by $\cover$) from the instance of $p$ to $\instI$.
      
      A path-pattern $p$ is {\em a cycle~(at $x$) in $\instI$} if this homomorphism can be chosen so that 
      the beginning and end variables are mapped to the same element~(to $x$) by $\cover$.
    \end{definition}
    Note that the path-patterns allow us to capture the notions of
    linear arc consistency and singleton linear arc consistency:
    \begin{itemize}
      \item LAC algorithm does not derive a contradiction on the instance $\instI$ 
        if and only if every path-pattern in $\instI$ has a solution;%
        \footnote{By a solution of a pattern we mean, of course, a solution of the underlying instance.}
      \item an instance $\instI$ is a SLAC instance if and only if
        for every variable $x$, every $a\in A_x$ and any path pattern $p$ which is a cycle at $x$ in $\instI$
        $p$ has a solution sending beginning and end variables of $p$ to $a$.
    \end{itemize}
    Indeed, a LAC algorithm derivation defines a path pattern~(and vice versa),
    and the path pattern has a solution if the derived set is non-empty.
    In the case of SLAC the algorithm does not shrink $A_x$ if fixing $x$ to $a\in A_x$ 
    and running LAC does not produce an empty set. 
    This is equivalent to the condition about path-patterns which are cycles at $x$.

    Throughout the paper we often create patterns by merging other patterns~(just like in the definition above).
    In such situations we assume that the variable sets of these patterns
    are disjoint and all the identifications of the variables are explicitly stated.
    E.g. if $p$ and $q$ are path-patterns on  disjoint sets of variables
    we define $p+q$ as a path-pattern created by identifying the end variable of $p$ with the beginning variable of $q$.
    The beginning variable of the sum is the beginning variable of $p$ and the end variable of the sum
    is the end variable of $q$.
    A pattern $-p$ is obtained from pattern $p$ by exchanging the functions of beginning and end variables.
  \subsection{Tree patterns}
    To define a tree-pattern we introduce a notion of an adjacency multigraph of an instance. 
    The vertex set of the adjacency multigraph of an instance consists of all the variables 
    and all the constraints of this instance,
    and we introduce one edge between a variable vertex and a constraint vertex 
    for every time the variable appears in the constraint.
    \begin{definition}
      An instance is {\em a tree~(forest) instance} if it's adjacency multigraph is a tree~(forest) 
      i.e. has no multiple edges and no cycles.

      {\em A tree-pattern} is a tree-instance with a selected variable called {\em the root}, 
      and a selected set of variables of degree one in the adjacency multigraph called {\em leaves}.
      
      A tree-pattern $p$ is {\em in the instance $\instI$} if there exists a homomorphism~%
      (usually denoted by $\cover$) from the instance of $p$ to $\instI$.
    \end{definition}
    Tree patterns are to arc consistency as path patterns to linear arc consistency i.e.
    \begin{itemize}
      \item AC checking algorithm does not derive a contradiction on an instance if and only if
        every tree pattern in this instance has a solution.
    \end{itemize}

  \subsection{Propagation via patterns}
    For a path-pattern~(tree-pattern) $p$ and a set $B\subseteq A$ we put $B+p$ to be the set consisting of all the values given
    to the end of $p$~(resp. root of $p$) by solutions of $p$ sending the beginning~(resp. leaf) variables  into $B$.
    
    Additionally, for path-patterns, by $B+p+q$ we mean $(B+p)+q$~(equivalently $B + (p+q)$) and by $B-p$ we mean $B+ (-p)$.

  \subsection{Universal covering tree instances}
      In order to capture the property ``AC does not derive a contradiction on $\instI$'' we 
      construct an auxiliary instance $\UCT(\instI)$. 
      We define it for connected instances first.
      
      We say that an instance is connected if it's adjacency multigraph is connected.
      For a connected instance $\instI$ the tree instances in $\instI$ form a Fra\"\i ss\'e family
      and we define the instance $\UCT(\instI)$ be its limit and $\cover$ to be the natural homomorphism
      from $\UCT(\instI)$ to $\instI$.
      
      Equivalently we can define $\UCT(\instI)$ to be the smallest~(usually countably infinite)
      instance such that:
      \begin{itemize}
        \item it is a tree instance, 
        \item it maps homomorphically onto $\instI$ via a map denoted by $\cover$,
        \item for every $(\tuple{x},\algC)$ constraint of $\instI$ there is $(\tuple{v},\algC)$ constraint of $\UCT(\instI)$
          such that $\cover(v_i)=x_i$ for every $i$,
        \item for every two variables $u,v$ of $\UCT(\instI)$ 
          if $\cover(u)=\cover(v)$ there is an automorphism $\psi$ of $\UCT(\instI)$ sending $u$ to $v$
          and such that $\cover(\psi(w))=\cover(w)$ for every $w$.
      \end{itemize}
      Note that each tree-pattern in $\instI$ maps homomorphically to $\UCT(\instI)$ and 
      the following conditions are equivalent:
      \begin{itemize}
        \item every tree-pattern in $\instI$ has a solution;
        \item $\UCT(\instI)$ has a solution;
        \item AC checking algorithm does not derive a contradiction on the instance $\instI$.
      \end{itemize}
      If $\instI$ is disconnected we take $\UCT(\instI)$ to be the union of $\UCT$'s of its connected components.

\section{Main theorem and sketch of the proof}\label{sect:result}
  In this section we restate Theorem~\ref{thm:mainsimple} and provide a more in-depth sketch
  of its proof using the algebraic notions from the previous sections.
  For a precise proofs of all the statements we refer the reader to the two last sections of the paper.

  Recall that all the templates are rigid cores, and that all the instances 
  have associated special constraints of the form $(x,A_x)$.
  In order to use the algebraic properties of the template  we restate Theorem~\ref{thm:mainsimple} as follows:
  \begin{theorem}\label{thm:main}
    Every SLAC instance over an \SDm template has a solution.
  \end{theorem}
  
  We start with a SLAC instance 
  $\instI$ over an \SDm template $\tempA$.
  As described in section~\ref{sect:mainsimple} our goal is to find algebras $\algA_x'\leq\algA_x$ for every
  $x$ and show that $\instI$ restricted to $\algA_x'$ is a SLAC instance.

  \subsection{There is no absorption in the instance}
    The first case we consider is when  none of the algebras $\algA_x$ has an absorbing subuniverse.
    In this case we choose an arbitrary $x$ with a non-trivial $\algA_x$ and choose a congruence $\alpha_x$ 
    such that $\algA_x/\alpha_x$ is simple.
    We pick an arbitrary block of this congruence to be $\algA'_x$.

    Take any path-pattern $p$ in $\instI$ such that $\cover$ maps beginning of $p$ to $x$ and let $w$ be the end variable of $p$.
    All the solutions to $p$, projected on the beginning and the end of $p$, produce an algebra 
    $\algR\subd\algA_x\times\algA_{\cover(w)}$ and we can quote $\algR$ on the first coordinate by $\alpha_x$ to obtain $\algR'$.
    If $\algR'$ is not-linked then it defines congruence $\alpha_{\cover(w)}$ on $\algA_{\cover(w)}$ and an isomorphism
    between $\algA_x/\alpha_x$ and $\algA_{\cover(w)}/\alpha_{\cover(w)}$ which forces the equivalence class of $\alpha_{\cover(w)}$
    that needs to be chosen as $\algA'_{\cover(w)}$ to be $\algA'_x +p$.

    One needs to show that the congruence and the isomorphism is independent on the choice of $p$; 
    and define $\algA'_y = \algA_y$ for $y$'s such that every $p$ defines a linked $\algR'$~%
    (note that in this case, by Theorem~\ref{thm:abs}, $\algR'=\algA_x/\alpha_x\times\algA_{\cover(w)}$).
    It remains to show that the restriction of $\instI$ to $\algA'_x$ is $1$-consistent and finally  a SLAC instance.
    This facts are proved in section~\ref{sect:noabs} using tools from section~\ref{sect:toolsnoabs}. 
  \subsection{There is absorption in the instance}
    In this case some algebra $\algA_{x'}$ has a non-trivial absorbing subuniverse, say $\algA'$.
    For every tree pattern $p$ in $\instI$ with all the leaf variables sent to $x'$ the algebra $\algA'+p$ is either empty
    or an absorbing subuniverse of $\algA_y$ where $y$ is the variable of $\instI$ to which $\cover$ sends the root of $p$.
    Using appropriate tree patterns we are able to define $\algA'_x\abs\algA_x$ for every $x$
    such that the restriction of $\instI$ to $\algA_x'$ is $1$-consistent.

    Next, using Theorem~\ref{thm:propagatetoabs}, taking for $\instI$ instances witnessing 
    big chunks of singleton linear arc consistency  we construct even smaller $\algA_x'$ such that the original instance 
    restricted to these algebras is a SLAC instance.
    The details of this proof are given in section~\ref{sect:abs}
    and use the tools established in the next section.

\section{Tools: absorbing subinstances}\label{sect:toolsabs}
  In this section we prove a theorem which plays a crucial role in the reduction in 
  the case with absorption.
  It does not require the template to be \SDm,
  in fact it requires no algebraic structure apart from the explicitly stated absorption.
  In our opinion this theorem is of independent interest and should find applications not connected to this paper.  

  In essence the theorem states that if a consistency can be found everywhere in the instance,
  then it can be found in any $1$-consistent absorbing subinstance.

  \begin{theorem}\label{thm:propagatetoabs}
    Let $\instI$ be an instance such that any $a\in\algA_x$ extends to a solution of $\instI$.
    Let, for every variable $x$,  $\algA'_x$ be an algebra such that
    \begin{enumerate}
      \item $\algA'_x\abs\algA_x$;
      \item $\UCT(\instI)$ can be solved with each variable $v$ in $\algA'_{\cover(v)}$.
    \end{enumerate}
    Then there exists a solution to $\instI$ with every variable $x$ evaluated into $\algA'_x$.
  \end{theorem}
  We proceed to prove  Theorem~\ref{thm:propagatetoabs}.
  Let a solution to $\UCT(\instI)$ be called {\em prime} if it sends every variable $v$ into $\algA'_{\cover(v)}$.
  For every variable $x$ of $\instI$ we choose an arbitrary $v$ of $\UCT(\instI)$ such that $\cover(v)=x$ and redefine $\algA'_x$
  to consist of all the values $v$ can take in prime solutions.
  Note that:
  \begin{itemize}
    \item the choice of $v$~(for a fixed  $x$)  does not matter~%
      (as $\UCT(\instI)$ has automorphisms moving potential candidates to each other) and
    \item the assumptions of the theorem still hold:
    \begin{enumerate}
      \item Fix an $x$ and $v$ such that $\cover(v)=x$. 
        For every $a\in \algA_x$ there is a solution $\sol$ of $\instI$ sending $x$ to $a$,
        therefore $\sol\circ\cover$ is a solution of $\UCT(\instI)$ sending $v$ to $a$.
        The algebra of prime solutions of $\UCT(\instI)$
        absorbs the algebra of all solutions of $\UCT(\instI)$ and therefore the new $\algA_x'$ absorbs $\algA_x$.
      \item Fix any prime solution of $\UCT(\instI)$; if a variable $v$ is mapped to $a$ then $a$ belongs to the new $\algA'_{\cover(v)}$
        which implies that the solution is, at every coordinate, in the new $\algA_x'$'s.
    \end{enumerate}
  \end{itemize}
  After such a change the restriction of $\instI$ to $\algA'_x$ is $1$-consistent. 

  Note that to prove the theorem it suffices to show that there exists 
  a prime solution to $\UCT(\instI)$ constant on $\cover^{-1}(x)$ for every variable $x$ of $\instI$.
  For induction assume that there is a prime solution to $\UCT(\instI)$ constant on each of the sets
  $\cover^{-1}(x_1),\dotsc,\cover^{-1}(x_k)$ and let $x$ be another variable.
  
  Let $\algE\leq\prod_{v\in\UCT(\instI)}\algA_{\cover(v)}$ be the algebra of all the solutions to $\UCT(\instI)$;
  for every $x, a\in\algA_x$ there is a solution $\sol$ such that $\sol(x) =a$;
  then $\sol\circ\cover$ sends $\cover^{-1}(x)$ to $a$ and we proved that 
  $\algE$ is subdirect in $\prod_{v\in\UCT(\instI)}\algA_{\cover(v)}$.

  Let $\algF\subd\prod_{v\in\UCT(\instI)}\algA'_{\cover(v)}$ be the algebra of prime solutions to $\UCT(\instI)$
  which is non-empty by the assumptions of our theorem.
  Clearly $\algF\abs\algE$ and we define $\algE', \algF'$ to be the subalgebras of $\algE, \algF$ respectively 
  consisting of solutions constant on each of $\cover^{-1}(x_1),\dotsc,\cover^{-1}(x_k)$.
  
  Since all the solutions to $\UCT(\instI)$ of the form $\sol\circ\cover$~(where $\sol$ is a solution to $\instI$)
  are constant on all the $\cover^{-1}(x)$'s the algebra $\algE'$ is still subdirect in $\prod_{v\in\UCT(\instI)}\algA_{\cover(v)}$.
  The algebra $\algF'$ is non-empty by the inductive assumption and, clearly $\algF'\abs\algE'$.

  \begin{claim}
    For any two variable $v,w$ of $\UCT(\instI)$ if $\cover(v) = \cover(w) = x$ then $\proj_v(\algF') = \proj_w(\algF')$.
  \end{claim}
  \begin{proof}
    Take any $a\in\proj_v(\algF')$ given by a solution $\sol$ and let $\psi$ be the automorphism of $\UCT(\instI)$ mapping $w$ to $v$.
    The solution $\sol\circ\psi$ is constant on $\cover^{-1}(x_i)$ by properties of $\psi$ and guarantees that $a\in\proj_w(\algF')$. 
  \end{proof}
  \begin{claim}
    Let $W$ be a finite set of variables of $\UCT(\instI)$ such that $\cover(W) = \{x\}$.
    There exists a tuple in $\algF'$ constant on $W$;
    equivalently there exists a prime solution to $\UCT(\instI)$ 
    constant on $\cover^{-1}(x_1),\dotsc,\cover^{-1}(x_k)$ and $W$.
  \end{claim}
  \begin{proof}
    By the previous claim $\algF'$ projects to the same algebra~(call it $\algB$) on every variable in $W$.
    Therefore $\proj_W(\algF')\subd\algB^{W}$ and $\proj_W(\algF')\abs(\proj_W(\algE')\cap\algB^W)$
    and the last algebra contains all the constants.
    By Proposition~\ref{prop:absconst} $\algF'$ contains a constant tuple.
  \end{proof}
  Since, for every finite $W\subseteq\cover^{-1}(x)$, we have a prime solution constant on  
  $\cover^{-1}(x_1),\dotsc,\cover^{-1}(x_k)$ as well as $W$, 
  the compactness argument provides a prime solution constant on all of the sets 
  $\cover^{-1}(x_1),\dotsc,\cover^{-1}(x_k),\cover^{-1}(x)$. 
  This finishes the proof of an inductive step and of Theorem~\ref{thm:propagatetoabs} as well.

\section{Tools: algebras with no absorption}\label{sect:toolsnoabs}
  In this section most of the algebras are simple, have no absorption and lie in an \SDm variety. 
  
  From~\cite{BW} we know the following basic fact concerning simple, absorption free algebras in \SDm varieties.
  \begin{proposition}\label{prop:old}
    Let $\algA_1,\dotsc,\algA_k$ be simple algebras with no absorbing subuniverses which lie in an \SDm variety.
    If $\algR\subd\prod_i\algA_i$ is such that $\pi_i\vee\pi_j = 1_{\algR}$ then $\algR =\prod_i\algA_i$.
  \end{proposition}
  \noindent The following corollary is an easy consequence of previous proposition.
  \begin{corollary}\label{cor:noabs}
    Let $\algA_1,\dotsc,\algA_k$ be simple algebras with no absorbing subuniverses which lie in an \SDm variety.
    If $\algR\subd\prod_i\algA_i$ and $\algR'\abs\algR$ then $\algR'=\algR$. 
  \end{corollary}
  \begin{proof}
    Consider $\proj_{i,j}(\algR)$. 
    If it is a graph of a bijection then, as $\proj_{i,j}(\algR')\subd\algA_i\times\algA_j$ is 
    a subalgebra of $\proj_{i,j}(\algR)$, the algebras $\proj_{i,j}(\algR')$ and $ \proj_{i,j}(\algR)$ have to coincide.
    
    If, on the other hand, $\proj_{i,j}(\algR)$ is not a graph of bijection then, in $\proj_{i,j}(\algR)$
    we have $\pi_1\vee\pi_2 = 1_{\proj_{i,j}(\algR)}$ and, by Theorem~\ref{thm:abs}, $\proj_{i,j}(\algR) = \algA_i\times\algA_j$.
    By Fact~\ref{fact:linkness} in $\proj_{i,j}(\algR')$ also $\pi_1\vee\pi_2 = 1_{\proj_{i,j}(\algR')}$
    and by the reasoning from the previous sentence $\proj_{i,j}(\algR') = \algA_i\times\algA_j$. 
    We conclude that in this case $\proj_{i,j}(\algR) = \proj_{i,j}(\algR')$ as well.
    
    If, for some $i,j$, the projection $\pi_{i,j}(\algR)$ is a graph of a bijection we can drop the $j$-th coordinate,
    i.e. substitute $\algR,\algR'$ with their projections to all but the $j$-th coordinate.
    Note that every tuple in such a projection extends uniquely to the $j$-th coordinate and thus to an
    element of the original algebra.

    After repeating the procedure sufficiently any times we obtain new $\algR,\algR'$ 
    such that $\pi_{i,j}(\algR) = \pi_{i,j}(\algR')= \algA_i\times\algA_j$ for every $i,j$. 
    These new algebras are full products~(by Proposition~\ref{prop:old}). 
    
    As every element of the new $\algR,\algR'$ uniquely extends to an element 
    of the original $\algR, \algR'$ respectively, the corollary is proved.
    Indeed every tuple $\tuple{a}$ of the projected $\algR'$ extends to a tuple 
    in the original $\algR'$, this extension needs to be the unique extension of $\tuple{a}$ to a tuple in $\algR$.
  \end{proof}
  \noindent The next corollary also follows from Proposition~\ref{prop:old}.
  \begin{corollary}\label{cor:nocon}
    Let $\algA_1,\dotsc,\algA_k$ be simple algebras with no absorbing subuniverses which lie in an \SDm variety.
    If $\alpha$ is a congruence on $\algR\subd\prod_i\algA_i$ such that for every $\tuple{a}\in\algR$  and every $i$
    $\proj_i(\tuple{a}/\alpha) = \algA_i$ then $\alpha=1_{\algR}$.
  \end{corollary}
  \begin{proof}
    Repeating a trick from the proof of Corollary~\ref{cor:noabs}
    we consider $\proj_{i,j}(\algR)$ and drop coordinate $j$~%
    (i.e. project on all coordinates different than $j$)
    whenever $\proj_{i,j}(\algR)$ is a graph of bijection.
    As every tuple on the coordinates other than $j$ extends uniquely to the $j$-th coordinate 
    the relation $\alpha$ inherited to the new $\algR$ remains transitive i.e. still is a congruence.

    After removing all such coordinates $\algR$ projects fully on every pair of coordinates by Theorem~\ref{thm:abs}. 
    Consider $\alpha$ as subuniverse of $\prod_i\algA_i\times\prod_i\algA_i$. 
    It projects fully on every pair of coordinates~(by assumption about $\alpha$)
    and thus by Proposition~\ref{prop:old} it is a full power which means that $\alpha$ is the full congruence.
  \end{proof}
  \begin{corollary}\label{cor:subd}
    Let $\algB, \algA_1,\dotsc,\algA_k$ be  algebras with no absorbing subuniverses which lie in an \SDm variety,
    such that for every $i$ the algebra $\algA_i$ is simple.
    If $\algR\subd\algB\times\prod_i\algA_i$~(i.e. $\algR$ is subdirect in $k+1$-ary product)
    and the congruence $\pi_1\vee\pi_i = 1_{\algR}$ for every $i$,
    then $\algR =\algB\times\proj_{2,\dotsc,k+1}(\algR)$.
  \end{corollary}
  \begin{proof}
    View $\algR$ as a subproduct of $\algB$ and $\proj_{2,\dotsc,k+1}(\algR)$.
    Since, for any $i$ $\pi_1\vee\pi_i = 1_{\algR}$ the $\proj_{1,i}(\algR)$ is linked
    and, by Theorem~\ref{thm:abs}, $\proj_{1,i}(\algR)= \algB\times\algA_i$.
    The congruence $\beta$, defined as in Fact~\ref{fact:linkness}, on $\proj_{2,\dotsc,k+1}(\algR)$ 
    satisfies the assumption of Corollary~\ref{cor:nocon} and therefore $\beta = 1_{\proj_{2,\dotsc,k+1}(\algR)}$.

    This implies that $\algR$~(still viewed as a subproduct of $\algB$ and $\proj_{2,\dotsc,k+1}(\algR)$)
    is linked i.e. satisfies assumptions of Theorem~\ref{thm:abs}.
    As neither $\algB$~(by assumption) or $\proj_{2,\dotsc,k+1}(\algR)$~(by Corollary~\ref{cor:noabs})
    has a proper absorbing subuniverse we conclude that $\algR =\algB\times\proj_{2,\dotsc,k+1}(\algR)$ as required.
  \end{proof}
  \begin{corollary}\label{cor:newloops}
    Let $\algB_1,\dotsc, \algB_n, \algA_1,\dotsc,\algA_k$ lie in an \SDm variety and let 
    each $\algA_i$ be simple with no absorbing subuniverses.
    Let $\algR\subd\prod_i\algB_i\times\prod_i\algA_i$~(subdirect as a $k+n$-ary product) be such that
    \begin{itemize}
      \item $\proj_{1,\dotsc,n}(\algR)$ has no absorbing subuniverses, and
      \item for each $i,j$ such that $1\leq i \leq n < j \leq n +k$ we have $\pi_i\vee\pi_j = 1_{\algR}$.
    \end{itemize}
    Then $\algR = \proj_{1,\dotsc, n}(\algR) \times \proj_{n+1,\dotsc,n+k}(\algR)$.
  \end{corollary}
  \begin{proof}
    For $n=1$ this is exactly Corollary~\ref{cor:subd}. 
    We let $n=2$ and will show that for any $i>2$ the algebra $\proj_{1,2,i}(\algR)$ is
    equal to $\proj_{1,2}(\algR)\times \algA_i$. 
    View $\proj_{1,2,i}(\algR)$ as a subproduct of $\proj_{1,2}(\algR)$ and $\algA_i$, 
    if $\alpha$, defined as in Fact~\ref{fact:linkness}, is equal to $1_{\proj_{1,2}(\algR)}$
    we obtained our goal. 
    In such a case we put $\algB = \proj_{1,2}(\algR)$ and the result follows by Corollary~\ref{cor:subd}.

    All the following paragraphs of the proof, except for the last one,
    are devoted to the proof of a contradiction in the case when $\alpha$ is smaller.
    In this case $\alpha$ is a congruence on $\algB = \proj_{1,2}(\algR)$
    and $\beta$~(provided by Fact~\ref{fact:linkness}) is $0_{\algA_i}$~(since $\algA_i$ is simple).
    Further we have: 
    \begin{itemize}
      \item $\algB/\alpha$ is isomorphic to $\algA_i$ and
      \item each $\alpha$ block projects fully onto $\algB_1$ and onto $\algB_2$~(indeed take 
      $\proj_{1,i}(\algR)\subd\algB_1\times\algA_i$; 
      it satisfies assumptions of Theorem~\ref{thm:abs} and neither $\algA_i$ nor 
      $\algB_1$ has a non-trivial absorbing subuniverses~(the last one since $\algB$ had none); 
      thus $\proj_{1,i}(\algR)=\algB_1\times\algA_i$ and similarly for $2$ instead of $1$).
    \end{itemize}

    Define an algebra $\algP\leq(\algB_1\times\algB_2)^3$:
    \begin{equation*}
      \tuple{a}\in\algP \textrm{ iff } \forall i (a_{2i-1},a_{2i})\in\algB \textrm{ and } (a_{2i+1},a_{2i})\in\algB
    \end{equation*}

    Fix $b\in\algB_1,\, b'\in\algB_2$ s.t. $(b,b')\in\algB$
    and let $\algP'\leq\algP$ consists of $\tuple{a}\in\algP$ such that $a_1=b$ and $a_6 = b'$.
    Let $\algP''$ be obtained from $\algP'$ by taking, for every tuple  $\tuple{a}\in\algP'$,
    tuple $((a_1,a_2)/\alpha,(a_3,a_2)/\alpha,\dotsc,(a_{5},a_{6})/\alpha)\in(\algB/\alpha)^5$.
    The algebra $\algP''$ can be viewed~(using isomorphism between $\algB/\alpha$ and $\algA_i$)
    as a subalgebra of $\algA_i^5$.

    Fix arbitrary $i<j$;
    in this paragraph we show that the projection $\pi_{i,j}(\algP'')$ is full, i.e. equal to $\algA_i^2$.
    Locate the position $i$-th~(we assume wlog that it $i$ is odd)  position and construct the sequence:
    \begin{equation*}
      (b,\dotsc,b_i,b'',b_i,b_{i+1},b_{i+2},\dotsc,b')
    \end{equation*}
    where $b''$ is such that $(b_i,b'')/\alpha$ is mapped, via the isomorphism between $\algB/\alpha$ and $\algA_i$,
    to $a$.
    Such a $b''$ can be found since each $\alpha$ class projects fully onto $\algB_1$.
    Applying similar construction to $j$-th coordinate and $a'$ instead of $a$ we obtain a tuple in $\algP'$ which
    produces a tuple in $\algP''$ which has $a$ on $i$-th an $a'$ on $j$-th coordinate.
    
    By the previous paragraph we can apply Proposition~\ref{prop:old} to $\algP''$
    and conclude that it is a full power of $\algA_i$.
    By construction of $\algP''$ and $\algP'$ this implies that every 
    $\alpha$-class, viewed as a subdirect subalgebra of $\algB_1\times\algB_2$ 
    defines, in Fact~\ref{fact:linkness}, the same linkness congruences $\alpha'$ and $\beta'$.
    And the same biejection between classes of $\alpha'$ and $\beta'$.

    Take a single $\alpha'$ class $\algB_1'$ and corresponding $\beta'$ class $\algB_2'$.
    Take an arbitrary $a\in\algA_i$, denote the corresponding $\alpha$-class in $\algB$ by $\algB'\subd\algB_1\times\algB_2$
    and let $\algB'' = \algB'\cap\algB_1'\times\algB_2'$. 
    The algebra $\algB''$ is subdirect
    in $\algB_1'\times\algB_2'$~(by the choice of $\algB_1',\algB_2'$ and the subdirectness of $\alpha$-classes)
    and $\pi_1\vee\pi_2 = 1_{\algB''}$ by previous paragraph.
    By Corollary~\ref{cor:MAS} we obtain $\algB_1''\abs\algB_1', \algB_2''\abs\algB_2'$ such that
    $\algB''_1\times\algB''_1\leq\algB''$.

    We will show that $\{a\}\abs\algA_i$. 
    Let $t$ be a term witnessing $\algB_2''\abs\algB_2'$.
    Let $a'\in\algA_i$ fix $(b,b')\in\algB_1''\times\algB_2''$ and $(b,b'')\in\algB$ such the $(b,b'')/\alpha$ is a class isomorphic to $a'$.
    Then $t((b,b'),\dots,(b,b'),(b,b''),(b,b'),\dotsc,(b,b'))\in\{b\}\times\algB_2''$ and, quoting via $\alpha$ and using isomorphism
    we get $t(a,\dotsc,a,a',a,\dotsc,a) = a$.
    This contradicts the fact that $\algA_i$ has no proper absorbing subuniverses.
    The contradiction implies that $\alpha$ cannot be smaller than $1_{\proj_{1,2}(\algR)}$ and this concludes the proof for $n=2$.

    The remaining part of the proof is by induction on $n$. 
    We apply the induction step to $\pi_{2,\dotsc,n+k}(\algR)$~%
    (the assumptions are clearly satisfied) 
    and then the step for $n=2$ with the same $B_1$ but 
    new $B_2$ equal to $\pi_{2,\dotsc,n}(\algR)$.
  \end{proof}

\section{There is no absorption in the instance}\label{sect:noabs}
  In this case none of the $\algA_x$'s in the SLAC instance have an absorbing subuniverse.
  Find a variable $x$ such that $\algA_x$ has more than one element 
  and find a congruence $\congr{x}$ so that $\algA_x/\congr{x}$ is simple.
  Fix an arbitrary block of the congruence $\alpha_x$ and denote it by $\algA_x'$. 
  
  Let $p$ be a path-pattern in $\instI$~(via $\cover$) such that $\cover$ maps 
  the beginning vertex of $p$ to $x$ and let $w$ denote the end variable of $p$. 
  The pattern $p$ is {\em non-proper} if $\algA_x' + p = \algA_{\cover(w)}$.
  \begin{claim}
    Let $p$ be a non-proper path-pattern.
    For any $a\in\algA_x$ we have $a/\congr{x} + p = \algA_{\cover(w)}$.
  \end{claim}
  \begin{proof}
    Take an algebra of all the solutions of $p$ and project it on the beggining and end variables to 
    obtain $\algR\subd\algA_x\times\algA_{\cover(w)}$~(subdirectness follows from $1$-consistency of $\instI$).
    Let $R' = \{(a/\congr{x},a'):(a,a')\in\algR\}$ and $\algR$ be the subalgebra of $\algA_x/\congr{x}\times\algA_{\cover(w)}$
    with universe $R'$.

    Since $p$ is non-proper $\algR'$ is linked and since neither $\algA_x/\congr{x}$ nor $\algA_{\cover(w)}$ has an absorbing subuniverse
    $\algR'$ is a full product  by Theorem~\ref{thm:abs} and the claim is proved.
  \end{proof}
  \noindent We call all the other patterns~(with beggining variable mapped to $x$) proper.
  \begin{claim}
    Let $p$ be a proper path-pattern.
    \begin{enumerate}
      \item for any $a\in\algA_x  : a/\congr{x} + p - p = a/\congr{x}$
      \item the relation $\congr{\cover(w)}$ defined by 
        \begin{equation*}
          b\congrrel{\cover(w)} b' \textrm{ iff } \exists\, a,a'\, : a\congrrel{x}a',\, b\in a +p,\, b'\in a' + p
        \end{equation*}
        is a congruence on $\algA_{\cover(w)}$,
      \item $\algA_x/\congr{x}$ is isomorphic to $\algA_{\cover(w)}/\congr{\cover(w)}$ 
        and the isomorphism is given by projecting all the solutions of $p$
        to begging and end variables and
      \item if $\cover(w)=x$ then for every $a'\in\algA_x$ $a'/\alpha_x + p = a'/\alpha_x$.
    \end{enumerate}
  \end{claim}
  \begin{proof}
    Define algebras $\algR, \algR'$ as in the proof of previous claim.
    If $\algR'$ is linked it is a full product by Theorem~\ref{thm:abs}~%
    (as neither $\algA_x/\congr{x}$ nor $\algA_{\cover(w)}$ has an absorbing subuniverse) 
    which contradicts the fact that $p$ is proper. 
    Therefore $\algR'$ is not linked and then $\alpha$, defined by Fact~\ref{fact:linkness} for $\algR'$, 
    is $0_{\algA_x/\congr{x}}$~(as $\algA_x/\congr{x}$ is simple).

    In this case $\beta$, defined by Fact~\ref{fact:linkness} for $\algR'$, is exactly $\congr{{\cover(w)}}$. 
    Let $\algR''\subd\algA_x/\alpha_x\times\algA_{\cover(w)}/\alpha_{\cover(w)}$
    be the algebra with a universe $R'' = \{(a/\congr{x},a'/\congr{\cover(w)}):(a,a')\in\algR\}$.
    The algebra $\algR''$ is a graph of an isomorphism and therefore everything, except item 4. is proved.

    To see item 4. let ${\cover(w)}=x$.
    Since $\instI$ is a SLAC instance, $a'/\alpha_x + p \supseteq a'/\alpha_x$
    and if the two sets are different $\algR'$ is linked and $p$ is non-proper which is a contradiction.
  \end{proof}

  Call a variable $y$ of $\instI$ {\em proper} if there exists a proper path-pattern~(with beginning mapped to $x$)  
  and $\cover(w)=y$.
  Denote by $\algA_y'$ the congruence block of $\alpha_y$ such that $\algA_x' + p =\algA_y'$.
  For every variable $y$ which is not a proper variable put $\algA'_y=\algA_y$;
  and let $\instI'$ be an instance obtained by restricting every variable $x$ of $\instI$ to $\algA'_x$.
  The following claim implies, among other things, that the set $\algA_y'$ does not depend on the choice of a proper pattern.
  \begin{claim}
    Let $y,z$ be proper variables and let $p$ be a path-pattern in~$\instI$ with the begining variable mapped to $y$ and the
    end variable mapped to $z$. 
    Let $p_y, p_z$ be proper patterns that define $\algA_y'$ and $\algA_z'$. 
    Then
    \begin{enumerate}
      \item either, for every $a\in\algA_y$, $a/\congr{y} + p = \algA_z$, and for every $a'\in\algA_z$, $a'/\congr{z} -p = \algA_y$ or
      \item for every $a\in\algA_x$, $a/\congr{x} + p_y + p = a/\congr{x} + p_z$ and $a/\congr{x} +p_z -p =a/\congr{z} + p_y$.
    \end{enumerate}
  \end{claim}
  \begin{proof}
    Consider algebras $\algR,\algR'$ defined for $p$ as in the proof of the previous claim.

    If $\algR'$ is not linked then, as $\algA_y/\congr{y}$ is simple, we got for every $a\in\algA_x$
    $a/\congr{x} + p_y + p - p - p_y = a/\congr{x}$.
    Moreove, since $p_z$ is a proper pattern, we also get $a/\congr{x} + p_z-p_z = a/\congr{x}$.
    This implies that $\algA'_z - p_z + p_y + p$ is smaller than $\algA_z$, and 
    using proof identical as a proof of item 4. from last claim we get for every $a\in\algA_z$,
    $a/\congr{z} - p_z + p_y + p = a/\congr{z}$. 
    This means the $\algR''$ defined for pattern $-p_z+p_y+p$ is a graph of the identity map on $\algA_z/\congr{z}$
    and we are in case 2.

    If $\algR'$ defined for $-p$ is not linked the reasoning is identical and we still end up in case 2.
    If the algebras $\algR'$ defined for $p$ and for $-p$ are both linked 
    then  both are full products and the condition of case 1. follows easily.
  \end{proof}
    The claim implies that, for a proper variable $z$, the congruence $\congr{z}$ and the choice
    of $\algA'_z$ does not depend on the choice of proper pattern.
    Indeed take a proper $z$ and $p_z$ a pattern used to define $\algA'_z$;
    let $p$ be a proper pattern with end mapped to $z$.
    We apply last claim with $x=y$~($p_y$ the 1-step pattern derived from constraint $(x,\algA_x)$) variable $z$ and pattern $p$.
    As $p$ is proper we are not in case 1 and in case 2 the conclusion is obvious.

  \subsection{The instance $\instI'$ is $1$-consistent}
    Let $(\tuple{x},\algC)$ with $\algC\subd\prod_i\algA_{x_i}$ be an arbitrary constraint of $\instI$.
    We fix an arbitrary coordinate, without loss of generality coordinate $1$,

    We denote the following property by $(\dag)$:
    every element in $\algA'_{x_1}$ extends to a tuple~(with this element at the first coordinate) 
    in $\algC$ which projects into
    $\algA'_{x_i}$'s on all the coordinates $i>1$.

    First we project $\algC$ on coordinate one together with all the coordinates $i$ such that $x_i$ is proper. 
    If the property $(\dag)$ holds for such a  projection of  $\algC$  
    it will hold for $\algC$ as well since for every non-proper variable $y$ we have $\algA'_y=\algA_y$.
    We call this projection $\algC$ as well.

    Let $C' = \{(a_1,a_2/\alpha_{x_2},\dotsc,a_k/\alpha_{x_k}):\tuple{a} \in \algC\}$ and  
    $\algC'$ be the induced algebra over universe $C'$. 
    If $x_1$ is not proper we have $\pi_j\vee\pi_1 = 1_{\algC'}$ for every $j$.
    If $x_1$ is proper then, by the last claim of previous section, we have two options
    \begin{itemize}
      \item either $\proj_{1,j}(\algC')$ is linked and therefore, by Theorem~\ref{thm:abs}, full product,
      \item or the projection $\proj_{1,j}(\algC')$ establishes a bijection between $\algA_{x_1}/\congr{x_1}$ 
        and $\algA_{x_j}/\congr{x_j}$ which maps $\algA'_{x_1}$ to $\algA'_{x_j}$;
        this means that every tuple with an element of $\algA'_{x_1}$ at first coordinate has an element of $\algA'_{x_j}$ on coordinate $j$
        and therefore, from the point of view of property $(\dag)$,  coordinate $j$ is not-important:
        property $(\dag)$ holds for $\algC$ if and only if it holds for the projection of $\algC$ 
        to coordinates different than $j$; in this case we substitute $\algC$ with such a projection.
    \end{itemize}

    We obtained an algebra which satisfies assumptions of Corollary~\ref{cor:subd}
    and thus is a full product of $\algA_{x_1}$ and $\proj_{2,\dotsc,k}(\algC')$.
    It remains to show that $\proj_{2,\dotsc,k}(\algC')$ contains a tuple fully in $\algA_{x_i}'/\alpha_{x_i}$'s.

    Again we consider all $\proj_{i,j}(\algC')$ for $i,j>1$. 
    If such a projection is a graph of bijection this bijection, by the last claim of previous section,
    maps $\algA'_{x_i}$ to $\algA'_{x_j}$ and it suffices to prove property $(\dag)$ for a projection of $\algC$ to
    the coordinates different than $j$.

    After projecting all possible coordinates we obtain $\algC$ such that every projection on two coordinates is 
    linked and therefore by Theorem~\ref{thm:abs} full.
    Such an algebra is, by Proposition~\ref{prop:old}, a full power and property $(\dag)$ trivially holds.
    This implies that every element of $\algA_{x_1}$ extends to a tuple fully in $\algA_x'$'s 
    which concludes the proof of $1$-consistency of $\instI'$.

  \subsection{$\instI'$ is a SLAC instance}
    Let $y$ be any variable and $p$ any pattern in $\instI$ from $y$ to $y$.
    Let $\algC\subd\algA_y^2\times\prod_i\algA_{x_i}$ be the algebra of all the solutions of this pattern.

    Since $\instI$ is SLAC
    $\proj_{1,2}(\algC)\supseteq \{(a,a):a\in\algA_y\}$.
    Our goal is to show that for any $a'\in\algA'_y$ there is a tuple $(a',a',a'_1,\dotsc,a'_k)\in\algC$ such that
    $a'_j\in\algA'_{x_j}$ for every $j$.
    Again we denote this property by $(\dag)$.
    Let $\algC'$ be $\algC$ quoted by $\congr{x_j}$ for any coordinate $j > 2$ such that $x_j$ is proper.

    Consider $i$ equal to $1$ or $2$ and $j>2$ such that $x_j$ is proper.
    If $y$ is non-proper the projection of $\pi_{i,j}(\algC')$ needs to be linked and therefore full.
    If $y$ is proper the projection $\pi_{i,j}(\algC')$ can be equivalently defined by a pattern~(since $\instI$ is $1$-consistent)
    and, similarly like in the previous section, it is either full product or defines a graph of bijection
    mapping $\algA_y'/\congr{y}$ to $\algA_{x_j}'/\congr{x_j}$. 
    In the latter case, similarly as in the previous subsection, it suffices to prove $(\dag)$
    for the projection of $\algC$ on coordinates different than $j$ -- we substitute $\algC$ with such a projection.
    
    Now let $\algC''$ be a projection of  $\algC'$ on the coordinates $1,2$ and all the coordinates $j$ such that $x_j$ is proper.
    Clearly it suffices to prove $(\dag)$ for $\algC''$.
    We take a minimal absorbing subuniverse $\algG$ of $\proj_{1,2}(\algC')$ and put
    $\algC''' = \{\tuple{a}\in\algC'' : \proj_{1,2}(\tuple{a})\in\algG\}$.
    As $\algC'''\abs\algC''$ the assumption of 
    Corollary~\ref{cor:newloops} are satisfied and we conclude that
    $\algC'''$ is a full product of $\algG$ and $\proj_{3,\dotsc,n}(\algC'')$.

    Note that $\algG\subd\algA_y^2$~(as $\algA_y$ has no absorbing subalgebras)
    and $\algG\abs\proj_{1,2}(\algC')$ which contains all the constant tuples.
    By Proposition~\ref{prop:absconst} $\algG$ contains constant tuples and,
    as the set of constant tuples in $\algG$ absorbs the set of constant tuples in $\algC'$
    and therefore defines an absorbing subuniverse of $\algA_y$,
    further all the constant tuples are in $\algG$.

    It remains to show that $\proj_{3,\dotsc,n}(\algC'')$ contains a tuple equal to $\algA'_{x_j}/\congr{x_j}$ for every $j$.
    Note that tuples in $\proj_{3,\dotsc,n}(\algC'')$ are given by solutions of a path-pattern.
    The $1$-consistency of $\instI'$~(which was proved in the previous subsection) provides required tuple.
    This shows that $\instI'$ is SLAC and the reduction is done in the case when there
    is no absorption in $\algA_x$'s.

\section{There is absorption in the instance}\label{sect:abs}

  In this case we assume that at least one $\algA_x$ has a non-trivial absorbing subuniverse.
  \subsection{There exists absorbing $1$-consistent subinstance of $\instI$}
    The first goal is to find a system of $\algA'_x$'s such that 
    \begin{itemize}
      \item $\algA'_x\abs\algA_x$ for every $x$ and
      \item $\UCT(\instI)$ can be solved with value of every $v$ in $\algA'_{\cover(v)}$~(again we call these solutions prime).
    \end{itemize}

    In order to find such $\algA'_x$ we define a preorder
    by putting $(x,\algB)\prop(x',\algB')$~($\algB\leq\algA_x$ and $\emptyset\neq\algB\neq\algA_x$ and the same for $x',\algB'$) if
    \begin{itemize}
      \item there is a tree pattern $p$ in $\instI$ with a homomorphism 
        mapping all the leaves to $x$ and the root to $x'$
      \item $\algB + p = \algB'$. 
    \end{itemize}
    We make the relation $\prop$ reflexive by adding required pairs and note that 
    and it is transitive by the ``addition'' of tree patterns which substitutes every leaf with a copy of a tree.
    
    Fix $\algA'$ and $x'$ such that $\algA'$ is a non-trivial absorbing subuniverse of $\algA_{x'}$,
    and restrict $\prop$ to the elements above the pair $(\algA',x')$.
    Note that, by $1$-consistency of instance $\instI$, for every $(y,\algB)$ in the restricted preorder $\algB\abs\algA_y$.
    Take the largest equivalence relation which is included in the restricted $\prop$, 
    and denote a maximal~(under $\prop$) equivalence class of this relation by $\lastcomponent$.

    We call a variable of $\instI$ {\em proper} if it appears in a pair in $\lastcomponent$,
    and postpone a proof of the following claim until later.
    \begin{claim}
      For every proper $x$, variable of $\instI$, there is an algebra $\algB$ such that:
      \begin{itemize}
        \item $(x,\algB)\in\lastcomponent$, and 
        \item for every $\algB'$ if $(x,\algB')\in\lastcomponent$ then $\algB\subseteq\algB'$.
      \end{itemize}
    \end{claim}
    For proper $x$'s we put $\algA'_x$ to be the unique minimal algebra provided by the previous claim,
    and for other $x$'s we put $\algA'_x = \algA_x$.

    Take an arbitrary constraint $(\tuple{x},\algC)$.
    By symmetry it suffices to show that it is subdirect with respect to the last coordinate. 
    Let $n$ be the largest number such that, 
    \begin{equation*}
      \proj_n\big(\proj_{1,\dotsc,n}(\algC)\cap\prod_{i\in \{1,\dotsc,n\}}\algA'_{x_i}\big) = \algA'_{x_n}
    \end{equation*}
    Clearly $n>1$ and if $n$ is smaller than the arity of $\algC$ define~(for contradiction) the following tree-pattern:
    \begin{enumerate}
      \item take the constraint $(\tuple{x},\algC)$~(make all variables of $\tuple{x}$ different) 
        and choose $n+1$'st variable of this constraint to be a root;
      \item fix any proper variable $x$ and, for each $i \leq n$, 
      \begin{enumerate}
        \item if $x_i$ is proper and $x_i\neq x$ 
          identify a root of new copy of the tree pattern witnessing $(x,\algA'_x)\prop(x_i,\algA'_{x_i})$
          with $x_i$;
        \item otherwise do nothing;
      \end{enumerate}
      \item define leaves of the tree pattern to consist of all the leaves introduced in 2.a~%
        (note that, as $x\neq x_i$, they are leaves of the new pattern),
        together with all the variables $x_i, i\leq n$ from $(\tuple{x},\algC)$ which are equal to $x$.
    \end{enumerate}
    Such a tree pattern is in $\instI$ and defines $(x,\algA'_x)\prop(x_{n+1},\algB)$ for $\algB$ such
    that $\algA'_{x_{n+1}}\not\subseteq\algB$ which implies that $x_{n+1}$ is proper
    and, at the same time, contradicts the claim.
    This clearly implies that $\UCT(\instI)$ has a prime solution and therefore all that remains is to prove the claim.
    
    The remaining part of this subsection contains a proof of the claim.
    Suppose, for a contradiction, that for some fixed $x$,
    there is more than one minimal~(under inclusion) set $\algB$ s.t.  $(x,\algB)\in\lastcomponent$.
    Let $\lastsets$ consists of $\algB$'s such that $(x,\algB)\in\lastcomponent$
    and let $\lastsets'$ consists of minimal under inclusion elements of $\lastsets$.

    Note that, for any $\algB\in\lastsets, \algB'\in\lastsets'$ either $\algB'\subseteq\algB$ or
    $\algB'\cap\algB = \emptyset$. 
    Indeed, suppose otherwise and let $p,q$ be tree-patterns in $\instI$ :
    $\algB'+p = \algB, \algB+q = \algB'$.
    Note that, by the assumption, roots of $p,q$ cannot be their leaves; therefore we can
    create a pattern by joining disjoint copies of $p$ an $q$ at the roots and
    adding  a new disjoint copy of $p$ at every leaf of $q$.
    This creates a tree-pattern in $\instI$ from $x$ to $x$ which defines $\algB'\cap\algB$ from $\algB$ --
    this contradicts the definition of $\algB'\in\lastsets'$.

    Now fix $\algB\in\lastsets$ such that:
    \begin{enumerate}
      \item there exists $\algB'\in\lastsets'$ such that $\algB\cap\algB'=\emptyset$
      \item $\algB$ is maximal, under inclusion, among elements of $\lastsets$ satisfying condition 1.
    \end{enumerate}
    Let $\algB'_1,\dotsc,\algB'_n$ be the elements of $\lastsets'$ which intersect empty with $\algB$.
    Fix tree patterns~(in $\instI$ from $x$ to $x$) $p_i$ such that $\algB'_i+p_i = \algB'_{i+1}$
    and patterns $q,q': \algB+q = \algB'_1, \algB'_n+q' =\algB$.
    Note that none of the above pattern patterns can have a root which is a leaf.

    Create a new pattern, call it $p$, iteratively:
    \begin{enumerate}
      \item start with $q'$;
      \item for $i$ from $n-1$ to $1$:
        for every leaf of the current pattern add a new copy of $p_i$ 
        and identify its root with this leaf;
      \item for every leaf of the current pattern add a copy of $q$ and identify its root with this leaf.
    \end{enumerate}
    Note that $\algB+p=\algB$, and after the iteration of the loop in step 2.
    for $i$ we got a pattern producing $\algB$ from $\algB'_i$.
    
    Next we iteratively modify $p$: take a leaf of $p$ and let $p'$
    be the same pattern as $p$ with the only difference that the fixed leaf of $p$
    is no longer tagged as a leaf in $p'$.
    If $\algB+p'=\algB$ we substitute $p$ for $p'$ and repeat the procedure.

    Fix any leaf of $p$, call it $v$, 
    and define $\algC\leq\algA_x^2$ to be a projection of all the solutions to $p$~%
    (without any restriction as to how the leaves are evaluated) to $v$ and the root.

    \begin{subclaim}\label{subclaim:c}
      The following hold:
      \begin{enumerate}
        \item $\algC$ is subdirect in $\algA_x^2$;
        \item for every $a\in\algA_x$ we have $(a,a)\in\algC$;
        \item for any $i$ and any $a\in\algB'_i$ we have $b\in\algB: (b,a)\in\algC$.
      \end{enumerate}
    \end{subclaim}
    \begin{proof}
      Item 1. holds by $1$-consistency of $\instI$. 
      For item 2. let $p'$ be a subpattern of $p$ which is a path-pattern connecting $v$ to the root of $p$. 
      We got $a\in\{a\}+p'$ by the fact that $\instI$ is SLAC. 
      As any solution to $p'$ extends to a solution of $p$ by $1$-consistency 2. is proved.

      For item 3. let $v'$ be the root of the pattern $p_{i-1}$~(or $q$ if $i=1$)
      on the path from $v$ to the root. 
      The path $p'$ from $v'$ to the root is a path pattern,
      and, similarly to item 2., we have $a\in \{a\}+p'$.
      The tree below $v'$ can be solved with root sent to $a$ and all its leaves in $\algB$~(by the construction of $p$).
      The solution to this tree can be glued with the solution to $p'$ and extended~(again by $1$-consistency of $\instI$)
      to a solution of $p$. 
      This shows fact 3.
    \end{proof}
    Let $\algC'\leq\algA_x^2$ be the same projection of $p$, but we require the solution
    of $p$ to send all the leaves different than $v$ to $\algB$.
    \begin{subclaim}\label{subclaim:cprime}
      The following hold:
      \begin{enumerate}
        \item $\algC'\abs\algC$;        
        \item $\algB + \algC' = \algB$;
        \item every $\algB'_i$ is a subset of $\proj_2(\algC')$;
        \item for every $i$ and every $a\in\algB'_i$ there exists $a'\in\bigcup_j\algB'_j$ such that $(a',a)\in\algC'$.
      \end{enumerate}
    \end{subclaim}
    \begin{proof}
      Item 1. follows directly from the fact that the set of solutions which defines $\algC'$ absorbs the set of all solutions to $p$.
      Item 2. is a straightforwrd consequence of the definition of $p$.
        
      For item 3. let $p'$ be a pattern identical to $p$, but with $v$ no longer a leaf in $p'$. 
      Either $\algB +p' = \algA_x$ or $\algB+p'\in\lastsets$. 
      In the first case item 3. is obvious, and in the second it follows from the maximality of $\algB$.
          
      For item 4. let $p'$ be as in the previous paragraph and 
      let $p''$ be obtained by identifying roots of disjoint copies of $p'$ 
      and a pattern witnessing $(x,\algB)\prop(x,\algB'_i)$.
      Note that by the previous item $\algB+p'' = \algB'_i$ and
      take all the solutions of $p''$ involved~(i.e. solutions which send all the leaves of $p''$ into $\algB$) 
      and set $\algB'$ to be the set of all the values these solutions take on $v$.
      
      Now, as $p''$ can be redefined with $v$ as a root, $\algB'$ is in $\lastsets$ or equal to $\algA_x$.
      Actually, by the construction, $\algB'\cap\algB= \emptyset$ and thus 
      we can take $\algB''\subseteq\algB'$ such that  $\algB''\in\lastsets'$. 
      Let $p'''$ the pattern obtained by identifying the root of the pattern witnessing $(x,\algB)\prop(x,\algB'')$
      with $v$ in $p''$.
      Now $\algB+p'''$ is a non-empty subset of $\algB'_i$ therefore it is in $\lastsets$, and by the minimality of $\algB'_i$
      it has to be equal to $\algB'_i$.
      This clearly implies a solution to $p'$ sending $v$ to $\algB''$ and root of $p'$ to $a$ 
      for any $a\in\algB'_i$.
      This 
      finishes the proof of item 4.
    \end{proof}
    Finally, putting all the things together, we choose an arbtrary $a\in\bigcup_i\algB'_i$ 
    and, using Subclaim~\ref{subclaim:cprime}.4, find $a'\in\bigcup_i\algB'_i$ such that $(a',a)\in \algC'$.
    By repeating this procedure we get $a\in\bigcup_i\algB'_i$ such that $(a,a)$ is included in $(\algC')^{(k)}$ 
    for some $k$.
    
    By Subclaim~\ref{subclaim:c}.3 we get $a'\in\algB$ such that $(a',a)\in\algC$,
    and, by a procedure similar to the one in previous paragraph, but with $\algB$ in place of $\bigcup_i\algB'_i$, 
    $a'$ in place of  $a$ and using Subclaim~\ref{subclaim:cprime}.2, we get 
    $a''\in \algB$ such that both $(a'',a'')$ is included in $(\algC')^{(k')}$
    and $(a'',a)$ in $(\algC)^{(k')}$.

    Let $\algD'$ be $(\algC')^{(kk')}$, and
    $\algD$ be $(\algC)^{(kk')}$.
    Clearly $(a,a),(a'',a'')\in\algD'$ while $(a'',a)\in\algD$;
    $\algB + \algD' = \algB$ and $\algD'\abs\algD$ and let 
    $t$ be the term witnessing this absorption.
    Consider the sequence 
    \begin{multline*}
      t(a'',\dotsc,a''), t(a,a'',\dotsc,a''),t(a,a,a'',\dotsc,a'')\dotsc,\\ \dotsc,t(a,\dotsc,a,a''),t(a,\dotsc,a).
    \end{multline*}
    Since $\algD'\abs\algD$ every two consecutive elements of this sequence create a pair in $\algD'$,
    but the whole sequence starts with $a''\in\algB$ and ends with $a\notin\algB$ which 
    contradicts the fact that $\algB+\algD'=\algB$.
    This finishes the proof of the claim which was the only missing part of this subsection.
  \subsection{There is a SLAC instance within $\algA'_x$'s}
    The following procedure produces a sequence of  instances and their homomorphisms into $\instI$.
    We denote the instances by 
    $(\instJ_1,\cover_1),\dotsc$:
    \begin{enumerate}
      \item take any path pattern from $x$ to $x$ in $\instI$ and obtain an instance $\instJ_1$
        by identifying the start and end variables of the pattern, 
        let $\cover_1$ be a homomorphism from $\instJ_1$ to $\instI$
      \item to obtain $(\instJ_{i+1},\cover_{i+1})$ from $(\instJ_i,\cover_i)$ 
        take any variable $x$ from $\instI$ and any path pattern $p$ from $x$ to $x$ in $\instI$ with homomorphism $\cover'$
        \begin{enumerate}
          \item if $x$ is not in the image $\cover$ let $\instJ_{i+1}$ be a disjoint union of $\instJ_i$
            and $p$ after identification of beginning and end variables, and $\cover_{i+1}$
            a disjoint union of $\cover_i$ and $\cover'$~(after the identification).
          \item if $x$ appears in the image of $\cover$, $\instJ_{i+1}$ is obtained from the instance obtained 
            in the previous case by identifying the beginning~(and end) variable from $p$ with one of the variables of $\instJ_i$
            mapped to $x$.
        \end{enumerate} 
    \end{enumerate}
    The variables $x$ in 2. and the identifications in 2.b can be chosen in such a way that
    \begin{itemize}
      \item eventually $\cover_i$ maps $\instJ_i$ onto $\instI$, 
      \item for every $i$, every $v$ in $\instJ_i$ and every path-pattern $p$ from $\cover_i(v)$ to $\cover_i(v)$ 
        the $p$ with begging and end identified with each other and with $v$ is a subpattern of some $\instJ_{i'}$.
    \end{itemize}

    Consider any $(\instJ_i,\cover_i)$ with $\algA_v$ defined to be $\algA_{\cover_i(v)}$
    and $\algA'_v = \algA'_{\cover(v)}$. 
    For every variable $v$ of $\instJ_i$ and every $a\in\algA_v$ the map $v\mapsto a$ extends to a solution of $\instJ_i$
    as $\instI$ is SLAC.
    Therefore, by Theorem~\ref{thm:propagatetoabs}, $\instJ_i$ has a prime solution i.e. a solution mapping 
    each $v$ to $\algA'_{\cover(v)}$.

    Every $\instJ_i$ has a prime solution and every $\instJ_j, j>i$ is an extension of $\instJ_i$.
    Therefore the set of values assigned to a variable $v$ by prime solutions of  $\instJ_i$ in the sequence 
    $\instJ_i,\instJ_{i+1},\dotsc$
    will eventually stabilise to some subset  of $\algA'_{\cover(v)}$.
    Moreover if $\cover(v)=\cover(v')$ the variables $v$ and $v'$ will stabilise, by the construction, at the same set.
    These sets can be taken for new $\algA_x$'s to create an instance smaller than $\instI$
    and satisfying singleton linear arc consistency.

\section{Acknowledgments}

  We thank Libor Barto and Jakub Bulin for many fruitful discussions and simplifications of some of the proofs in the paper.
  We thank anonymous reviewers for pointing out some inconsistencies in the presentation and other comments on improving 
  readability of the paper.

\bibliographystyle{plain}

\bibliography{2016}

\begin{thebibliography}{10}

\bibitem{FOCSBW}
L.~Barto and M.~Kozik.
\newblock Constraint {S}atisfaction {P}roblems of bounded width.
\newblock In {\em Foundations of Computer Science, 2009. FOCS '09. 50th Annual
  IEEE Symposium on}, pages 595--603, Oct 2009.

\bibitem{LICScyclic}
L.~Barto and M.~Kozik.
\newblock New conditions for {T}aylor varieties and {CSP}.
\newblock In {\em Logic in Computer Science (LICS), 2010 25th Annual IEEE
  Symposium on}, pages 100--109, July 2010.

\bibitem{BWColapse}
Libor Barto.
\newblock The collapse of the bounded width hierarchy.
\newblock {\em Journal of Logic and Computation}, 2014.

\bibitem{RobustSTOC}
Libor Barto and Marcin Kozik.
\newblock Robust satisfiability of {C}onstraint {S}atisfaction {P}roblems.
\newblock In {\em Proceedings of the Forty-fourth Annual ACM Symposium on
  Theory of Computing}, STOC '12, pages 931--940, New York, NY, USA, 2012. ACM.

\bibitem{BW}
Libor Barto and Marcin Kozik.
\newblock Constraint {S}atisfaction {P}roblems solvable by local consistency
  methods.
\newblock {\em J. ACM}, 61(1):3:1--3:19, January 2014.

\bibitem{nuLinDATA}
Libor Barto, Marcin Kozik, and Ross Willard.
\newblock Near unanimity constraints have bounded pathwidth duality.
\newblock In {\em Proceedings of the 2012 27th Annual IEEE/ACM Symposium on
  Logic in Computer Science}, LICS '12, pages 125--134, Washington, DC, USA,
  2012. IEEE Computer Society.

\bibitem{SACCompl}
Christian Bessiere and Romuald Debruyne.
\newblock Theoretical analysis of singleton arc consistency and its extensions.
\newblock {\em Artificial Intelligence}, 172(1):29 -- 41, 2008.

\bibitem{Bodirsky2010}
Manuel Bodirsky and Hubie Chen.
\newblock Peek arc consistency.
\newblock {\em Theoretical Computer Science}, 411(2):445 -- 453, 2010.

\bibitem{GaloisOLD}
V.~G. Bodnar{\v{c}}uk, L.~A. Kalu{\v{z}}nin, V.~N. Kotov, and B.~A. Romov.
\newblock Galois theory for {P}ost algebras. {I}, {II}.
\newblock {\em Kibernetika (Kiev)}, (3):1--10; ibid. 1969, no. 5, 1--9, 1969.

\bibitem{BulBW}
Andrei Bulatov.
\newblock Bounded relational width.
\newblock 2009.
\newblock manuscript.

\bibitem{duality2}
Andrei Bulatov, Peter Jeavons, and Andrei Krokhin.
\newblock Classifying the complexity of constraints using finite algebras.
\newblock {\em SIAM J. Comput.}, 34:720--742, March 2005.

\bibitem{duality1}
Andrei~A. Bulatov, Andrei~A. Krokhin, and Peter Jeavons.
\newblock Constraint satisfaction problems and finite algebras.
\newblock In {\em Automata, languages and programming (Geneva, 2000)}, volume
  1853 of {\em Lecture Notes in Comput. Sci.}, pages 272--282. Springer,
  Berlin, 2000.

\bibitem{SAC}
Hubie Chen, Victor Dalmau, and Berit Gru{\ss}ien.
\newblock Arc consistency and friends.
\newblock {\em Journal of Logic and Computation}, 23(1):87--108, 2013.

\bibitem{majorityLinDATA}
Victor Dalmau and Andrei Krokhin.
\newblock Majority constraints have bounded pathwidth duality.
\newblock {\em European Journal of Combinatorics}, 29(4):821 -- 837, 2008.
\newblock Homomorphisms: Structure and Highlights Homomorphisms: Structure and
  Highlights.

\bibitem{DKRobust}
Victor Dalmau and Andrei Krokhin.
\newblock Robust satisfiability for {CSP}s: Hardness and algorithmic results.
\newblock {\em ACM Trans. Comput. Theory}, 5(4):15:1--15:25, November 2013.

\bibitem{DalmauWidth1}
Victor Dalmau and Justin Pearson.
\newblock Closure functions and width 1 problems.
\newblock In Joxan Jaffar, editor, {\em Principles and Practice of Constraint
  Programming - CP{'}99}, volume 1713 of {\em Lecture Notes in Computer
  Science}, pages 159--173. Springer Berlin Heidelberg, 1999.

\bibitem{Debruyne97somepracticable}
Romuald Debruyne and Christian Bessiere.
\newblock Some practicable filtering techniques for the {C}onstraint
  {S}atisfaction {P}roblem.
\newblock In {\em In Proceedings of IJCAI'97}, pages 412--417, 1997.

\bibitem{FV98}
Tom{\'a}s Feder and Moshe~Y. Vardi.
\newblock The computational structure of monotone monadic {SNP} and constraint
  satisfaction: A study through {D}atalog and group theory.
\newblock {\em SIAM Journal on Computing}, 28(1):57--104, 1998.

\bibitem{RobustConj}
Venkatesan Guruswami and Yuan Zhou.
\newblock Tight bounds on the approximability of almost-satisfiable {H}orn
  {SAT} and exact hitting set.
\newblock In Dana Randall, editor, {\em SODA}, pages 1574--1589. SIAM, 2011.

\bibitem{algebraicBW}
Benoit Larose and L{\'a}szl{\'o} Z{\'a}dori.
\newblock Bounded width problems and algebras.
\newblock {\em Algebra universalis}, 56(3-4):439--466, 2007.

\end{thebibliography}

\appendix

\section{Weaker notion of consistency}
  In this section we will show that that
  a consistency notion strictly weaker then SLAC works for all the CSPs of bounded width.
  This weaker notion has an advantage of working well with SDP relaxations of CSPs.
  
  \begin{definition}
    An arc-consistent instance is a \pq instance if, 
    for every variable $x$, element $a\in A_x$ and every two circles $p,q$ at $x$,
    there exists $j$ such that $a\in a+ j(p+q) +p$.
  \end{definition}
  \noindent Using this definition we can state an analogue of Theorem~\ref{thm:main}:
  \begin{theorem}\label{main:pq}
    Every \pq instance over an \SDm template has a solution.    
  \end{theorem}
  \noindent A proof of Theorem~\ref{main:pq} is very similar to the proof of Theorem~\ref{thm:main}.
  This section contains only a list of changes which are necessary to make the former proof work for \pq consistency.
  
  In the remaining part of the section we work with a \pq instance $\inst I$.
  We will find, for every variable $x$, an algebra $\alg A_x'\leq\alg A_x$
  and show that the instance $\inst I$ restricted to the new sets is a \pq instance as well.
  Before we proceed with the proof we need to verify the following easy claim:
  \begin{claim}\label{claim:commonj}
    In the definition of \pq instance the phrase ``there exists $j$'' can be substituted with 
    ``there exists $k$ such that for every $j$ greater than $k$''. 
    In particular we can find a single $j$ which works for all the $a\in A_x$.
  \end{claim}
  \begin{proof}
    Consider the sequence $A_0=\{a\}, A_1 = A_0 +(p+q), A_2= A_1 +(p+q),\dotsc$. 
    First we show that this sequence needs to stabilise. 
    Suppose, for a contradiction, that we have two sets $A',A''$ such that 
    $A' + l(p+q)= A''$ and $A'' + l'(p+q) = A'$.
    Without loss of generality we assume that there is $a'\in A'\setminus A''$.
    By definition of \pq instance, applied for $a'$ and two patterns $l(p+q)$ and $l'(p+q)$, 
    we get $a'\in A''$ which is a contradiction.

    Let $A'$ be the set that stabilises the previous sequence.
    The fact that  $a\in A'$ follows from the fact that $a$ belongs to infinitely many sets in the sequence,
    which is consequence of the definition of \pq instance for $a$,  the  pattern $(p+q)$ and the empty pattern.
    
    It remains to prove that $A'+p=A'$:
    let $A'' = A'+p$
    so that $A'' +q =A'$. 
    If $A'\neq A''$ then an element in $A'\setminus A''$ directly contradicts the definition
    of a \pq instance and an element in $A''\setminus A'$ contradicts the definition with roles of $p$ and $q$ exchanged.
  \end{proof}

  In the remaining part of these section we freely use results from 
  section 7 and section 8 as they do not require a context of a CSP instance.
  The remaining two section deal with two cases of the which can appear: 
  either some  $\alg A_x$'s has a proper absorbing subuniverse none of the do.

  \subsection{There is no absorption in the instance~(analogue of section 9)}
    In the main part of section 9 all the claims hold
    and the only difference is in the proof of item 4 of claim 2: 
    we cannot assume that $a'/\alpha_x +p \supseteq a'/\alpha_x$.
    Instead w use item 3 and note that if the isomorphism sends $a'/\alpha_x$ to a class that does not contain $a'$
    then $a'$ with the patterns $p$ and $-p$ contradict the definition of a \pq instance.

    The subsection 9.1 remains unchanged. 
    In the subsection 9.2. in it is shown that the smaller instance in a SLAC instance,
    in here we will show that it is a \pq instance:
    we fix an element $a\in A_y$ and two patterns $p',q'$ from $y$ to $y$.
    By claim~\ref{claim:commonj} we find $j$ such that for all $a'\in A_y$ we have $a'\in \{a'\} + j(p'+q') +p'$ and follow 
    the same proof as in section 9.2 with $j(p'+q')+p'$ taken in place of $p$.

  \subsection{There is absorption in the instance~(analogue of section 10)}\label{sect:pqabs}
    The modification in this case are more subtle than in the previous one.
    Section 10.1 contains a single claim and two subclaims.
    In order to prove that single claim we need to work with a different pattern $p$%
    ~(which is defined right before subclaim 1).

    Fix $p$ and $v$~(a leaf of $p$) to be like in the last paragraph before subclaim 1.
    Denote by $r$ the path pattern from $v$ to the root of $p$.
    The path pattern $r$ has a number of initial fragments that start and end in $x$~%
    (the root of $q$ and the roots of $p_i$'s provide such initial fragment).
    Choose $j$ to be such that for every $a\in A_x$ and for every such initial fragment $r'$ of $r$
    $a\in\{a\}+jr +r'$~%
    (this can be done by Claim~\ref{claim:commonj}).
    The proof in section 10.1 continues with a new pattern taken for $p$.
    The new pattern is obtained by joining $j+1$ copies 
    of $p$ at each step identifying leaf $v$ with the root of $p$~%
    (to obtain a slender tree/comb of trees corresponding to $p$).
    The new $v$ of such pattern is the single $v$ not joined to a root of $p$ 
    and the root of this pattern is the single root of $p$ not joined to any $v$.

    Subclaim 1 follows, by the reasoning identical to the one in section 10.1, 
    from the construction of the new pattern $p$.
    In subclaim 2 items 1 and 2 hold directly from definitions.
    Items 3 and 4 hold for the original $p$ as their proofs do not depend on the SLAC assumption.
    Fortunately item 4 for the original $p$ immediately implies item 4 for the modified $p$ as the new $\alg C'$ 
    is essentially the old $\alg C'$ composed with itself $j+1$ times.
    This immediately implies item 3 for the new $p$ and the subclaim 2 holds.
    The rest of the subsection 10.1 requires no modifications.

    In subsection 10.2 an instance is constructed by adding path-patterns which are circles at $x$.
    The original reasoning stands because of two important properties of such patterns in any SLAC instance:
    \begin{enumerate}
      \item if all the circle patterns at $x$ have solutions sending $x\mapsto a$~%
        (in the original 10.2 the solutions were in the smaller instance)
        then the part of the definition of a SLAC holds for $x$ and $a$;
      \item for any variable $y$ in such a circle pattern and any $b\in A_y$
        the circle pattern can be solved~%
        (in the original 10.2 the solution was in the bigger instance)
        in such a way that $y\mapsto b$.
    \end{enumerate}

    In our modification we start with $a\in A_x$ and two patterns $p$ and $q$.
    Instead of a plain $p$ circle~(as in the original proof) we will consider 
    a $j(p+q)+p$ circle for some $j$~(to be determined later).
    Note that no matter what $j$ we choose the analogue of item 1. holds:
    finding a solution to our circle pattern guarantees no problems with $a$ in 
    the definition of \pq instance.

    In the remaining part of this section we will find $j$ such that $j(p+q) + p$
    satisfied item 2.
    By claim~\ref{claim:commonj} we can find $j$ such that for every $a\in A_x$
    $a\in \{a\} + j(p+q) +p$, but for other variables in the pattern we need to tweak it.

    Pick a variable $y$ in $p+q$ and say that it splits $p$ in two parts i.e. $p'+p'' = p$ 
    and $p'' + q + p'$ is a path pattern from $y$ to $y$.
    Find $j'$ such that, for every $b\in A_{y}$ the set $\{b\}+j'(p''+q+p')$ contains $b$ and is stabilised
    and, at the same time $\{b\} -j'(p''+q+p')$ contains $b$ and is stabilised~%
    (this can be done by applying claim~\ref{claim:commonj} twice).
    
    Now if $j$~(which we are attempting to choose) is greater than $2j'$ then every copy of $y$
    in the pattern $j(p+q)+p$ is either further than $j'$ steps of the form $(p''+q+p')$
    from the end or from the beginning.
    Say it is $j''\geq j'$ steps from the end~(in the other case we would use $-(p''+q+p')=-p'-q-p''$)
    then, for every $b$ we have $\{b\} + j''(p''+q+p')$ stabilised.
    But, by the same reasoning as in claim~\ref{claim:commonj}, this implies that $A$ defined as $\{b\} + j''(p''+q+p') +p''$ 
    is a subset of $A_x$ such that $A+p = A$ and $A+q=A$. 
    Thus we have
    $A = \{b\} + j''(p''+q+p') +p'' + p' + p''$.
    If this particular copy of $y$ was $j'''$ steps of the form $(p''+q+p')$ from the beginning the 
    set $\{b\} + j''(p''+q+p') +p'' + p' + j'''(p''+q+p')$ is still the stabilised set i.e. it contains $b$.
    Thus the  item 2 was proved for all the copies of $y$.
    To finish the proof of item 2. in  general we take $j$ large enough so it would work for every variable in $p+q$.
    The remaining part of the proof is identical as the original.

  \section{A weak analogue of Theorem~\ref{thm:propagatetoabs}}
    \noindent In this section we state and prove a slightly weaker version of Theorem~\ref{thm:propagatetoabs}:
    \begin{theorem}\label{thm:propagatetoabsmore}
      Let $\instI$ be an instance such that any $a\in\algA_x$ extends to a solution of $\instI$.
      Let $\instI'$ be an {\em arc-consistent} instance on the same set of variables as $\inst I$ such that:
      \begin{enumerate}
        \item for every variable $x$ we have $\algA'_x\abs\algA_x$~(where $\algA'_x$ is from $\instI'$ and $\algA_x$ from $\instI$),
        \item for every constraint $((x_1,\dotsc,x_n),R')$ in $\instI'$ there is a corresponding constraint $((x_1,\dotsc,x_n),R)$ in
          $\instI$ with $\algR'\abs\algR$.
      \end{enumerate}
      Then $\inst I'$ has a solution.
    \end{theorem}
    \noindent A proof of Theorem~\ref{thm:propagatetoabsmore} is identical to the proof of Theorem~\ref{thm:propagatetoabs} and
    we leave it to the reader.
    One consequences of this theorem is that we can slightly generalize the reasoning from section 10.2 to obtain:
    \begin{corollary}
      Let $\instI$ be a \pq instance .
      Let $\instI'$ be an {\em arc-consistent} instance on the same set of variables as $\inst I$ such that:
      \begin{enumerate}
        \item for every variable $x$ we have $\algA'_x\abs\algA_x$~(where $\algA'_x$ is from $\instI'$ and $\algA_x$ from $\instI$),
        \item for every constraint $((x_1,\dotsc,x_n),R')$ in $\instI'$ there is a corresponding constraint $((x_1,\dotsc,x_n),R)$ in
          $\instI$ with $\algR'\abs\algR$.
      \end{enumerate}
      Then $\inst I'$ has a \pq subinstance.
    \end{corollary}

\end{document}